\documentclass[acmsmall]{acmart}

\AtBeginDocument{%
  }

\setcopyright{cc}
\setcctype{by}
\acmDOI{10.1145/3763099}
\acmYear{2025}
\acmJournal{PACMPL}
\acmVolume{9}
\acmNumber{OOPSLA2}
\acmArticle{321}
\acmMonth{10}

\usepackage{subcaption}
\usepackage{pifont}
\usepackage{listings}
\usepackage{listings-rust}
\usepackage{xspace}
\usepackage{xparse}
\usepackage{xstring}
\usepackage[dvipsnames]{xcolor}
\usepackage{tabularray}
\usepackage{enumerate} 
\usepackage[ff,mm,sets,adversary,advantage,asymptotics,probability,keys,primitives,operators,lambda,logic,notions]{cryptocode}
\usepackage[T2A,T1]{fontenc}
\usepackage{hyperref}
\usepackage[flushleft]{threeparttable}
\usepackage{cuted}
\usepackage{boxedminipage}
\usepackage{mathtools}
\usepackage{algorithm}
\usepackage{algorithmicx}
\usepackage[noend]{algpseudocode}
\usepackage{multirow}
\usepackage{amsmath,amsfonts,amsthm}
\usepackage{graphicx}
\usepackage{textcomp}
\usepackage{svg}

\newcommand{\protocol}[3][0.7\columnwidth]{
	\begin{boxedminipage}[t]{#1}
		\begin{center}
		\scriptsize{\textbf{#2}}
		\end{center}
		\procedure[mode=text, 
		linenumbering,
		codesize=\footnotesize]{}{
			#3
		}
	\end{boxedminipage}
}

\theoremstyle{definition}

\newtheorem{theorem}{Theorem}
\newtheorem{lemma}{Lemma}

\def\Snospace~{\S{}}

\graphicspath{{./figures/}}
\DeclareGraphicsExtensions{.pdf,.jpeg,.png}

\makeatletter
\let\old@lstKV@SwitchCases\lstKV@SwitchCases
\def\lstKV@SwitchCases#1#2#3{}
\makeatother
\usepackage{lstlinebgrd}
\makeatletter
\let\lstKV@SwitchCases\old@lstKV@SwitchCases

\lst@Key{numbers}{none}{%
    \def\lst@PlaceNumber{\lst@linebgrd}%
    \lstKV@SwitchCases{#1}%
    {none:\\%
     left:\def\lst@PlaceNumber{\llap{\normalfont
                \lst@numberstyle{\thelstnumber}\kern\lst@numbersep}\lst@linebgrd}\\%
     right:\def\lst@PlaceNumber{\rlap{\normalfont
                \kern\linewidth \kern\lst@numbersep
                \lst@numberstyle{\thelstnumber}}\lst@linebgrd}%
    }{\PackageError{Listings}{Numbers #1 unknown}\@ehc}}
\makeatother

\usetikzlibrary{positioning,shadows.blur,fit}
\definecolor{auburn}{rgb}{0.43, 0.21, 0.1}
\definecolor{codegreen}{rgb}{0,0.6,0}
\definecolor{codegray}{rgb}{0.5,0.5,0.5}
\definecolor{codepurple}{rgb}{0.58,0,0.82}
\definecolor{linkcolor}{rgb}{0.7,0,0}
\definecolor{Gray}{gray}{0.85}

\hypersetup{colorlinks = true, urlcolor = blue, linkcolor = linkcolor, citecolor = linkcolor}
 
\lstdefinelanguage
   [x64]{Assembler}     
   [x86masm]{Assembler} 
   {morekeywords={CDQE,CQO,CMPSQ,CMPXCHG16B,JRCXZ,LODSQ,MOVSXD, %
                  POPFQ,PUSHFQ,SCASQ,STOSQ,IRETQ,RDTSCP,SWAPGS, %
                  rax,rdx,rcx,rbx,rsi,rdi,rsp,rbp, %
                  r8,r8d,r8w,r8b,r9,r9d,r9w,r9b, %
                  r10,r10d,r10w,r10b,r11,r11d,r11w,r11b, %
                  r12,r12d,r12w,r12b,r13,r13d,r13w,r13b, %
                  r14,r14d,r14w,r14b,r15,r15d,r15w,r15b}} 

\lstdefinestyle{codestyle}{
  basicstyle=\footnotesize\ttfamily,
  numbers=none,
  escapeinside={<@}{@>},
}
\newcommand{\name}{\textsc{Agora}\xspace}

\newcommand{\confllvm}{ConfLLVM\xspace}
\newcommand{\confverify}{ConfVERIFY\xspace}

\newcommand{\changes}[1]{\ignore{{[#1]}}}
\newcommand{\xchanges}[1]{#1}
\newcommand{\ndsschanges}[1]{#1}

\newcommand{\spchanges}[1]{#1}

\newif\ifexample\examplefalse

\begin{document}

\title{Agora: Trust Less and Open More in Verification for Confidential Computing}

\author{Hongbo Chen}
\orcid{0000-0001-9922-4351}
\affiliation{%
  \institution{Indiana University Bloomington}
  \city{Bloomington}
  \country{USA}
}
\email{hc50@iu.edu}

\author{Quan Zhou}
\orcid{0009-0003-3497-7848}
\affiliation{%
  \institution{The Pennsylvania State University}
  \city{University Park}
  \country{USA}
}
\email{qfz5074@psu.edu}

\author{Sen Yang}
\orcid{0000-0002-8866-2097}
\affiliation{%
  \institution{Yale University}
  \city{New Haven}
  \country{USA}
}
\email{sen.yang@yale.edu}

\author{Sixuan Dang}
\orcid{0000-0002-3241-9530}
\affiliation{%
  \institution{Duke University}
  \city{Durham}
  \country{USA}
}
\email{sixuan.dang@duke.edu}

\author{Xing Han}
\orcid{0009-0004-9907-0988}
\affiliation{%
  \institution{Hong Kong University of Science and Technology}
  \city{Hong Kong}
  \country{China}
}
\email{xhanbg@connect.ust.hk}

\author{Danfeng Zhang}
\orcid{0000-0003-1942-6872}
\affiliation{%
  \institution{Duke University}
  \city{Durham}
  \country{USA}
}
\email{danfeng.zhang@duke.edu}

\author{Fan Zhang}
\orcid{0000-0002-8525-4514}
\affiliation{%
  \institution{Yale University}
  \city{New Haven}
  \country{USA}
}
\email{f.zhang@yale.edu}

\author{XiaoFeng Wang}
\orcid{0000-0002-0607-4946}
\affiliation{%
  \institution{Nanyang Technological University}
  \city{Singapore}
  \country{Singapore}
}
\email{xiaofeng.wang@ntu.edu.sg}

\begin{abstract}
\spchanges{Confidential computing (CC), designed for security-critical scenarios, uses remote attestation to guarantee code integrity on cloud servers.
However, CC alone cannot provide assurance of high-level security properties (e.g., no data leak) on the code.
In this paper, we introduce a novel framework, \name, scrupulously designed to provide a trustworthy and open verification platform for CC.}
To prompt trustworthiness, we observe that certain verification tasks can be delegated to untrusted entities, while the corresponding (smaller) validators are securely housed within the trusted computing base (TCB).
Moreover, through a novel blockchain-based bounty task manager, it also utilizes crowdsourcing to remove trust in complex theorem provers. These synergistic techniques successfully ameliorate the TCB size burden associated with two procedures: binary analysis and theorem proving. To prompt openness, \name supports a versatile assertion language that allows verification of various security policies. Moreover, the design of \name enables untrusted parties to participate in any complex processes out of \name's TCB.
By implementing verification workflows for software-based fault isolation, information flow control, and side-channel mitigation policies, our evaluation demonstrates the efficacy of \name.
\end{abstract}




\maketitle


\definecolor{myblue}{RGB}{70,130,180}
\definecolor{mydeepblue}{RGB}{65,105,225}
\definecolor{myburgundy}{RGB}{110,10,30}
\definecolor{mygreen}{RGB}{0,105,148}
\definecolor{mygrey}{RGB}{180, 180, 200}
\definecolor{idealfun}{RGB}{165,42,42}
\definecolor{check}{RGB}{11,141,10}
\definecolor{auburn}{rgb}{0.43, 0.21, 0.1}
\definecolor{codegreen}{rgb}{0,0.6,0}
\definecolor{codegray}{rgb}{0.5,0.5,0.5}
\definecolor{codepurple}{rgb}{0.58,0,0.82}
\definecolor{myviolet}{rgb}{0.58,0,0.82}
\definecolor{mypink}{RGB}{227,115,131}

\newcommand{\fanz}[1]{{\textcolor{purple}{[FZ: #1]}}}
\newcommand{\danfeng}[1]{{\textcolor{check}{DZ: #1}}}
\newcommand{\hongbo}[1]{{\textcolor{blue}{HC: #1}}}
\newcommand{\quan}[1]{{\textcolor{codepurple}{QZ: #1}}}
\newcommand{\sen}[1]{{\textcolor{olive}{SY: #1}}}
\newcommand{\done}{\textcolor{red}{\checkmark}}

\algnewcommand\SmallBlock{\small}

\newcommand{\BBH}{\text{BBH}\xspace}
\newcommand{\BTM}{\text{BTM}\xspace}
\newcommand{\skbtm}{\mathsf{sk_{BTM}}\xspace}
\newcommand{\pkbtm}{\mathsf{pk_{BTM}}\xspace}
\newcommand{\addrbtm}{\mathsf{addr_{BTM}}\xspace}
\newcommand{\addrbbh}{\mathsf{addr_{BBH}}\xspace}
\newcommand{\addrbbhsc}{\mathsf{addr_{BBH}^{SC}}\xspace}
\newcommand{\addrbbhbtm}{\mathsf{addr_{BBH}^{BTM}}\xspace}
\newcommand{\basicreward}{\mathsf{R_{basic}}}
\newcommand{\bugreward}{\mathsf{R_{bug}}}
\newcommand{\query}{\textproc{query}}
\newcommand{\determine}{\textproc{checkIsGenuine}}
\newcommand{\validate}{\textproc{validate}}
\newcommand{\mapping}{\textproc{mapping}}
\newcommand{\verifySignature}{\textproc{verifySig}}
\newcommand{\verifyAttestation}{\textproc{verifyAttest}}
\newcommand{\attest}{\textproc{attest}}
\newcommand{\pay}{\textproc{pay}}
\newcommand{\update}{\textproc{update}}
\newcommand{\unsat}{\texttt{UNSAT}\xspace}
\newcommand{\sat}{\texttt{SAT}\xspace}
\newcommand{\btmhash}{H_\text{BTM}}
\newcommand{\bundlehash}{H_\text{B}}
\newcommand{\functionhash}{H_\text{F}}
\newcommand{\blocktime}{\mathsf{time_{block}}}
\newcommand{\smartcontract}{\text{smart contract}\xspace}
\newcommand{\verifiedtasks}{\mathsf{verifiedTasks}}
\newcommand{\functions}{\mathsf{functions}}
\newcommand{\updatedfunctions}{\mathsf{updatedFunctions}}
\newcommand{\buggytasks}{\mathsf{buggyTasks}}
\newcommand{\taskbundles}{\mathsf{taskBundles}}

\newcommand{\bnf}[1]{\langle {\tt {#1}} \rangle}
\newcommand{\addr}[1]{\ifmmode{\tt 0x#1}\else ${\tt 0x#1}$\fi}
\newcommand{\choice}{\; \mid \;}
\newcommand{\bugbountyhunter}{BBH}
\newcommand{\mytilde}{\raise.17ex\hbox{$\scriptstyle\mathtt{\sim}$}}

\newcommand{\ignore}[1]{\iffalse{#1}\fi}

\newcommand{\parhead}[1]{\vspace{5pt}\noindent {\it {#1}.}}

\NewDocumentCommand{\vc}{o}{%
  \IfValueTF{#1}{%
    \IfStrEqCase{#1}{%
      {first}{\textit{assertion\xspace}}%
      {plural}{assertions\xspace}%
      {s}{assertions\xspace}%
      {initcap}{Assertion\xspace}%
    }[\textbf{Unknown option.\xspace}]
  }{
    assertion\xspace}%
}

\NewDocumentCommand{\vcgen}{o}{%
  \IfValueTF{#1}{%
    \IfStrEqCase{#1}{%
      {first}{\textit{\vc[first]generator\xspace}}%
      {plural}{\vc~generators\xspace}%
      {s}{\vc~generators\xspace}%
      {initcap}{\vc[initcap] Generator\xspace}%
    }[\textbf{Unknown option.}]
  }{%
    \vc generator\xspace%
  }%
}

\NewDocumentCommand{\regulation}{o}{%
  \IfValueTF{#1}{%
    \IfStrEqCase{#1}{%
      {first}{\textit{obligation\xspace}}%
      {plural}{obligations\xspace}%
      {s}{obligations\xspace}%
      {initcap}{Obligation\xspace}%
    }[\textbf{Unknown option.}]
  }{%
    obligation\xspace%
  }%
}

\section{Introduction}
\label{sec:intro}

\def\ccstory{\true}
\ifdefined\ccstory{


\spchanges{
Confidential computing (CC), powered by hardware Trusted Execution Environments (TEEs), has gained widespread adoption in CPUs and GPUs across modern cloud platforms.
However, their security primitives fail to address users' evolving demand in protecting data-in-use.
While remote attestation provides cryptographic proof of an TEE program's integrity through hash validation, it fails to establish meaningful security assurance (e.g., the CC application meets data confidentiality requirements).
Consequently, users are forced to place trust in a cryptographic hash value rather than being presented with concrete evidence proving the TEE program's security properties.
}

Static code analysis and program verification serve as the bedrock 
of the development of reliable and secure software, providing a defensive shield for critical system components. For example,
general verification frameworks have been developed and employed to verify diverse security policies~\cite{angr,kirchner2015framac,cruanes2013etb}.
However, the implicit trust placed in these heavyweight tools is against the design goal of TEEs: to reduce the Trusted Computing Base (TCB). Blindly trusting the verification frameworks also presents significant security risks. For example,
the verification stack, comprising program analyses and theorem provers, has exhibited various vulnerabilities despite extensive testing.
Sophisticated analyses have been reported with errors, such as the adversarial examples~\cite{sun2024validating} and vulnerabilities (e.g., CVE-2023-2163, CVE-2024-45020) in the Linux eBPF verifier.
Similarly, SMT solvers have faced both soundness issues~\cite{mansur2020detecting, park2021generative} and security vulnerabilities, as evidenced by CVE-2020-19725 and CVE-2024-37794.
%

\spchanges{Moreover, the monopolistic entities in the \emph{closed} ecosystem (with program analyses, theorem provers and so on) prevent contributions from third parties. For example, when a research group develops a new SMT solver that surpasses well-known solvers for certain kinds of constraints, how can they quickly earn the trust of users to participate in the program verification ecosystem?
Those issues are also seen in
application stores like Apple App Store and Google Play, which integrate a review process for the applications (e.g., security and privacy policy compliance)~\cite{AppRevie20, Preparey94}.}
Previous research has highlighted issues in closed verification ecosystems maintained by prominent companies~\cite{lalaine, koch2022keeping}.

One alternative approach is to develop lightweight verification tools tailored for \emph{specific} security policies, which usually bear a smaller TCB, plus some level of independence on closed ecosystems (e.g., via implementing a customized solver).
However, they are incapable of adapting to new policies due to policy disparities, such as underlying assumptions and domain-specific proving techniques and solvers.
For example, as recent developments in TEE support paradigms such as virtual machine and WebAssembly, we can also reuse existing policy-specific verifiers (e.g., VeriWASM~\cite{veriwasm} for software-based fault isolation). But the extensibility limitation persists: VeriWASM~\cite{veriwasm} is built on a tailored abstract interpretation for software-based fault isolation (SFI). Hence, adapting it to new policies, such as information flow control (IFC), or new proving techniques, such as symbolic execution, is infeasible.
Consequently, achieving comprehensive security—typically requiring the verification of multiple policies—often involves stacking disparate policy-specific verifiers. Unfortunately, this practice can lead to an inflated TCB, increasing system complexity and potentially introducing new vulnerabilities.


In this paper, we present \name, a verification framework that strikes the best features of existing approaches. Compared with general verification frameworks, \name is a \emph{more open} verification service that requires \emph{less trust} from users.
``Open more'' implies that an inclusive verification ecosystem welcomes participation from \emph{any} individual.
``Trust less'' implies that trust can be removed from the verification ecosystem participants, including the dictatorial parties in existing ecosystems.
Meanwhile, unlike policy-specific verifiers, the service also maintains compatibility with versatile policies and provides a \emph{comparable} TCB.

Achieving both trustworthiness and openness in one system is challenging:
supporting the verification of diverse policies inherently conflicts with ``trust less''.
General verification frameworks often include a series of static analyzers (e.g., alias analysis and dependency analysis), proving techniques (e.g., abstract interpretation and symbolic execution) and SMT solvers as their trusted components, significantly enlarging the TCB. To address this challenge, 
\begin{itemize}
    \item 
    \name separates verification into \emph{generating} policy-related facts from a program, as well as \emph{checking} the correctness of those facts and the overall security.
    This separation enables the delegation of complex program analyses to the untrusted domain, thereby reducing the TCB size.
    The trusted domain retains only a lightweight validator responsible for checking those untrusted facts against both the target program's semantics and policy requirements.
    This was made possible by incorporating an interfacing language which is simple but expressive enough to encode a variety of policy specifications.
    \item 
    \spchanges{\name involves a novel smart contract-based approach for theorem proving.}
    This system crowdsources theorem proving tasks and validates the results through a protocol to gauge the participants' contribution to verification (\autoref{sec:btm}).
    This design removes constraint solving from the TCB, leaving only a minimal component that checks the validity of a constraint solution from untrusted parties.
    Additionally, smart contract naturally allows verification results to be audited by end users, providing the required transparency.
\end{itemize}
}\else{
\remarks{Below is the original version of the introduction. We are merging the CC use scenario into this version as one candidate of the new introduction section.}

Static code analysis and program verification serve as the bedrock of developing reliable and secure software, providing essential safeguards for security-sensitive applications.
\ignore{\quan{... for security sensitive use scenarios such as confidential computing? And then expand with a few sentences about CC here?}}
\spchanges{
Various verification techniques are widely employed to prove high-level security properties, such as data confidentiality and controlled information flow.
A prominent example is confidential computing (CC), an emerging paradigm that leverages hardware-based Trusted Execution Environments (TEEs) for secure data processing on the cloud.
While TEEs ensure the integrity of the enclave program by checking measurement hashes via remote attestation,
the lack of high-level security guarantees in CC ecosystem can be resolved by a verification service.
}

\spchanges{
General program verification frameworks have been developed and employed to verify diverse security policies~\cite{angr,kirchner2015framac,cruanes2013etb}.
However, the implicit trust placed in these verification tools presents significant security risks.
The verification stack, comprising program analyzers and SMT solvers, has exhibited various vulnerabilities despite extensive testing.
SMT solvers have faced both soundness issues~\cite{mansur2020detecting, park2021generative} and security vulnerabilities, as evidenced by CVE-2020-19725 and CVE-2024-37794.
Similarly, sophisticated analyzers have also led to errors, such as the adversarial examples~\cite{sun2024validating} and vulnerabilities (e.g., CVE-2023-2163, CVE-2024-45020) in the Linux eBPF verifier.
%
}
\ignore{General verification frameworks, such as~\cite{binsec, bap}, are 
designed to verify a wide range of safety and security properties. Yet, in order to establish trust in the verified software, end users have to \emph{blindly trust} the correctness of the frameworks, analyses/verifiers built on top of them, as well as well-known SMT solvers underpinning the frameworks. The unconditional trust is usually justified by the fact that those components are developed by reputable parties and are well-tested. 
However, for example, even well-tested SMT solvers are prone to numerous bugs, including soundness issues~\cite{mansur2020detecting, park2021generative, bringolf2022finding, winterer2020unusual, winterer2020validating}.}
\ignore{
Moreover, the monopolistic entities in the \emph{closed} ecosystem prevent contributions from third parties. For example, when a research group develops a new SMT solver that surpasses well-known solvers for certain kinds of constraints, how can they quickly earn the trust of users to participate in the program verification ecosystem?
}
\ignore{Those issues are also seen in
application stores like Apple App Store and Google Play, which integrate a review process for the applications (e.g., security and privacy policy compliance)~\cite{AppRevie20, Preparey94}.}

\ignore{
Previous research has highlighted issues in closed verification ecosystems maintained by prominent companies~\cite{koch2022keeping}.
One natural solution is to provide a more open verification framework, embracing the intelligence from the general public.
\quan{I think it should be okay to directly go into the next paragraph.}
}


One alternative approach is to develop lightweight verification tools tailored for \emph{specific} security policies \spchanges{or application scenarios.
They usually bear smaller TCB and are more trustworthy.}
However, they face significant limitations in adaptability and integration with other policies or scenarios due to their specificities, including underlying assumptions and domain-specific proving techniques.
For example, VeriWASM~\cite{veriwasm} verifies software-based fault isolation (SFI) compliance via abstract interpretation,
but adapting it to new policies, such as information flow control (IFC), or alternative analysis techniques, such as symbolic execution, requires substantial redevelopment effort.
\spchanges{
Consequently, achieving comprehensive security—typically requiring the verification of multiple policies—often involves stacking disparate verifiers. Unfortunately, this practice can lead to an inflated TCB, increasing system complexity and potentially introducing new vulnerabilities.
}

\spchanges{
Recognizing the critical need to \textit{trust less} in verification, we introduce \name.
Our core design principle is to delegate complex, error-prone tasks to untrusted entities while maintaining lightweight validators within the TCB.
This approach leverages a fundamental insight in verification: the complexity of program analysis and constraint solving far exceeds that of result validation.
By shifting the complex logic burden outside the TCB and realizing straightforward results validation, \name effectively reduces the attack surface within the verification process.
\name successfully combines the advantages of both general-purpose and policy-specific verifiers while introducing a unique benefit: its architecture inherently supports a \textit{more open} ecosystem.
As validators process untrusted inputs securely, \name enables external contributors—including developers and researchers—to participate in program analysis and constraint-solving without compromising security integrity.
\ignore{In light of this observation, we put forth \name, a verification framework that strikes the best features of existing verification systems.
Compared with general verification frameworks, \name is a \emph{more open} verification service that requires \emph{less trust} from users.
``Open more'' implies that an inclusive verification ecosystem welcomes participation from \emph{any} individual.
``Trust less'' implies that trust can be removed from the verification ecosystem participants, including the dictatorial parties in existing ecosystems.
Meanwhile, unlike policy-specific verifiers, the service also maintains compatibility with versatile policies and provides a \emph{compariable} TCB.}
}
\ignore{
\quan{To delete} Such design allows manual inspection of the source code and the verification process and makes formal verification of the system easier in the future.}
\ignore{These goals are particularly important to confidential computing (CC), where controlling TCB size is vital~\cite{intel:sgxexplained}.}
\ignore{
\hongbo{to delete this paragraph} In this paper, we put forth a \emph{trustworthy, open, versatile, and auditable} binary verification service that addresses the shortcomings of existing tools and ecosystems \xchanges{and provide a new infrastructure for confidential computing}.
Firstly, trustworthiness is anchored in the foundations of program analysis and verification. 
Specifically, the verification service should \xchanges{require less trust from users} and possess a compact codebase with straightforward logic, enabling the majority of the service to undergo manual inspection, and even formal verification in the foreseeable future.
Secondly, we strive to construct an open verification ecosystem that encourages contributions from entities with varied technical backgrounds toward code verification.
Thirdly, a compact codebase should not limit the service's extensibility to different policies.
\changes{Fourthly, the verification service should never steal or expose proprietary algorithms in the binary.}
Lastly, the verification service must be publicly auditable, allowing everyone to inspect the verification process independently.
These features are especially vital in confidential computing, where remote attestation alone cannot build trust in enclave programs \changes{and software vendors keep the binary with IP confidential}, as it only verifies the enclaves' identity~\cite{intel:sgxexplained}, rather than its functionality. 
}

The dual objectives of trustworthiness and openness present significant challenges in verification system design.
Supporting the verification of diverse policies inherently conflicts with ``trust less''.
General verifiers often include a series of static analyzers (e.g., alias analysis and dependency analysis), proving techniques (e.g., abstract interpretation and symbolic execution), and SMT solvers as trusted components, significantly enlarging the TCB. To tackle these,
\begin{itemize}
    \item \spchanges{\name introduces a strategic bifurcation.} It separates verification into \emph{generating} policy-related facts from a program, as well as \emph{checking} the correctness of those facts and the overall security.
    This separation enables the delegation of complex program analyses to the untrusted domain (\autoref{sec:proof-checking}), thereby reducing the TCB size.
    The trusted domain retains only a lightweight validator responsible for checking those untrusted facts against both the target program's semantics and policy requirements.

\item 
\spchanges{\name involves a novel smart contract-based approach for theorem proving.}
This system crowdsources theorem proving tasks and validates the results through a protocol to gauge the participants' contribution to verification (\autoref{sec:btm}).
This design removes constraint solving from the TCB, leaving only a minimal component that checks the validity of a constraint solution from untrusted parties.

\end{itemize}

\ignore{
\hongbo{to delete this paragraph} In pursuit of the goals above,
\hongbo{may need a new name} we introduce \name, a novel binary verification service.
As a service, it conducts the verification pipeline while at the same time publishing verification results with evidence for users to query and audit.
A meticulously crafted proof system (\autoref{sec:proof-checking}) enables \name to delegate binary analysis to the untrusted domain, accepting versatile policies while shrinking the TCB size.
On the other hand, smart contracts are coupled with trusted execution environments (TEE) (\autoref{sec:btm}) to remove trusted third parties and outsource constraint solving tasks to untrusted participants, encouraging experimental algorithms but requiring less trust via trustworthy infrastructures.
Auditable verification results are documented and published via the blockchain.
}
}
\fi

To exemplify the adaptability of \name, we integrate three policies \spchanges{
that are particularly relevant to CC, having established counterparts in prior work~\cite{sgx:occlum, pobf, sgx:tsgx, sgx:deflection}}: software-based fault isolation (SFI), information flow control (IFC), and mitigation for load value injection (LVI) attack~\cite{lvi} (\autoref{sec:casestudy}). 
Notably, this integration only adds 625, 335, and 77 lines of policy-specific code, respectively.
\name maintains a TCB size comparable with each policy-specific verifier but significantly smaller than existing general verification frameworks.
\spchanges{
\name aligns with the principles in CC: TCB minimization and transparency enhancement.
Additionally, it alleviates the requirements for code disclosure to users and the burden on users to replicate complex verification procedures.
This improvement in usability and security makes \name well-suited for CC applications.}

In summary, \name make the following contributions:
\begin{itemize}
    \item The design of a binary verifier that offloads complex analysis to untrusted entities (\autoref{sec:proof-checking}) and supports versatile policies (\autoref{sec:casestudy}).
    To the best of our knowledge, \name is the first open and trustworthy verification framework that enables participation from any party without bloating TCB size.
    
    \item A bug bounty protocol that decentralizes constraint solving and promotes auditability (\autoref{sec:btm}). The protocol validates results from diverse but untrusted solving methods, providing auditable evidence to users.
    
    \item The implementation of \name and policy verifiers on top of it accommodating different policies (\autoref{sec:implementation}). We will release the code once the paper is accepted.

    
    \item An evaluation of \name's reduction in TCB size, performance enhancements, and estimated cost. (\autoref{sec:evaluation}).
\end{itemize}

\section{Background}
\label{sec:background}




\subsubsection*{Software-Based Fault Isolation}
SFI is a security mechanism devised to curb harmful behavior emanating from misbehaving code~\cite{efficientsfi}.
Particularly, by earmarking safe externals and a fault domain that encapsulates data and code territories, SFI regulates that memory reads/writes occur within the data region and control flow transfers only land in the code region or safe externals~\cite{sfi:principles}.
SFI can be enforced by dynamic binary translation or an inlined reference monitor,
the latter is often favored due to its efficiency, versatility, and verifiability~\cite{sfi:efficient, sfi:portable, sfi:strato}.
WebAssembly (WASM) operates within sandboxes.
Some code generators such as Lucet compile WASM code into native binaries to enhance performance, producing binaries fortified with SFI.

\subsubsection*{Information Flow Analysis}
Information flow analysis tracks private information and data derived from it in a computer system and prevents it from unintentionally leaking to the public domain. 
Various information flow analysis tools have been developed~\cite{myers1999jflow, pottier2002information, johnson2015exploring}. Information flow analysis assumes partial or complete labeling on the secrecy of variables, with the goal of proving that information marked or inferred to be secret cannot affect the values of public variables, formalized as the non-interference property~\cite{goguen1982security}.

\subsubsection*{Smart Contract and Blockchain}
A blockchain is a decentralized ledger that chronicles every transaction that transpired across its network.
This ledger is replicated and distributed among all network participants.
Smart contracts are programs housed within the blockchain, set to execute automatically once certain preset conditions are met.
They leverage the fault tolerance and high availability afforded by the blockchain, permitting individual nodes to interact with the deployed smart contracts and verify their states.
Systems like Ethereum~\cite{wood2014ethereum} have showcased the versatility of smart contracts across various applications.

\subsubsection*{Bug Bounty Program}
Software vendors reward individuals who identify bugs or vulnerabilities in their products.
Stimulated by the bounty reward, programmers and hackers often become bug bounty hunters (BBHs), incentivized to submit bug and vulnerability reports.
Upon confirmation of these bugs by the companies' engineering teams, BBHs receive their rewards.
Large vendors like Google and Microsoft manage their bounty programs, while others may entrust the program to third-party platforms like HackerOne~\cite{hackerone}.



\section{Related Work}
\label{sec:related}


\subsubsection*{Program Verification}
Despite rigorous formalization and careful implementation, verification tools are often built on diverse building blocks, such as program analysis frameworks and SMT solvers.
Efforts to formalize communication interfaces within verification 
pipelines~\cite{angr,kirchner2015framac,cruanes2013etb} have emerged, recognizing the need for modularization in evolving verification landscapes that integrate diverse components with versatile techniques.
Such interfaces can be realized by formal proof formats~\cite{wetzler2014drat,chihani2017semantic}, policy specifications~\cite{schneider2000enforceable,aktug2008conspec}, and/or intermediate representations of semantics~\cite{lattner2004llvm,dullien2009reil}.
However, existing frameworks assume a \emph{closed ecosystem} where all components involved in the pipeline are \emph{trusted}, resulting in large TCB sizes.

One exception is the seminal work of proof-carrying code (PCC), introduced by Necula and Lee~\cite{pcc-original}, which exemplifies a systematic approach wherein code producers accompany machine code with a formal proof of safety properties~\cite{necula1996safe,colby2000certifying, pcc:authentication, pcc-foundational, pcc-foundational:hoare, pcc-foundational:open}.
This verification paradigm removes trust on the formal proof generated by the code producer;
so the code consumer is only required to check the correctness of the accompanied proof, which is usually significantly easier than generating one (e.g., loop invariants).
\ndsschanges{PCC and \name share the insight of outsourcing complex tasks outside TCB but checking them inside TCB. One difference is that the code consumer in PCC needs to run a verification condition generator (VCGen)\footnote{Foundational PCC~\cite{pcc-foundational, pcc-foundational:hoare, pcc-foundational:open} shifts trust on VCGen to foundational mathematical logic, and hence, further reduces TCB of conventional PCC. However, it only applies to type-based PCC instantiations.}, which often consists of nontrivial program analyses for nontrivial policy. For example, one instantiation of PCC, the Cedilla Java compiler~\cite{colby2000certifying} contains a trusted VCGen with 23 KLOC (thousand lines of code). In comparison, the counterpart in \name, called assertion generator, is out of TCB (see~\autoref{sec:workflow}).
}
Moreover, PCC's implementation is typically tied to specific policies and even the kind of program analysis technique being used, hindering the verification of versatile policies in one integrated system.
\ignore{
\quan{A few other angles to look at: 1) The burden of proving policy compliance in the PCC paradigm is still on the software authors; 2) Although it happens only once at load time, the proof checking induces overhead, especially with two things to consider under CC: a) there may be more than one policy to be verified and b) each run of the program may need individual load (e.g., CC in FaaS). The above should be sufficient(?) to motivate an improved design of PCC, which would allow outside proof contributors (this might be contrived?), and allow for a true verify-once-use-many-times scenario (with the help of blockchain techniques).} \quan{Even for foundational-PCC, the underlying math logic would require the policy-compliance proofs to be encoded into logic formulae/constraints, which in turn requires a \textit{trusted} backend checker/solver to verify safety. Our design removes such needed trust.}
}

\ignore{
\quan{\textit{below is old version}} While PCC is aligned with our goal of trust less and open more\quan{this is not the same claim as later}, it still falls short in a few important aspects.
First, the consumer-side proof checkers still need to run nontrivial code analysis, such as Floyd-style verification condition generation~\cite{floyd1993}. 
For example, in the Cedilla Systems implementation of PCC~\cite{colby2000certifying}, the proof checker contains 23K\quan{clarify this, 16K VCGen + smaller checker} lines of C code, without counting the size of the backend SMT solvers.
Second, PCC lacks a common interface for varying policies. In fact, PCC's applications are typically tied to specific policies and even the kind of program analysis technique being used, hindering the verification of versatile policies in one integrated system. 
For instance, both the code producer and consumer need to compute Floyd-style verification conditions to verify safe assembly extensions of ML programs~\cite{pcc-original}.
Replacing the proof generation methods on the producer side, with say symbolic execution, fails to validate the proof. Third, the openness of PCC is limited: as proofs are written into the binary, proof generation is limited to only the code producer instead of any untrusted participant; relying on well-known SMT solvers also prohibits the participation of untrusted participants. \quan{\textit{end of old version}} 
}

\ignore{
BAP~\cite{bap} is a notable binary analysis platform that facilitates the creation of custom analyses via a domain-specific language. 
Another tool, angr~\cite{angr}, specializes in symbolic execution.
Besides, renowned platforms dedicated to reverse engineering and binary analysis, such as Radare2~\cite{radare2}, IDA Pro~\cite{idapro}, and Ghidra~\cite{ghidra}, have made significant contributions \quan{vague, add what contributions}\hongbo{delete this sentence?} to malware analysis.
Notably, these analyzers often possess a large TCB due to the intricacies of code analysis.
In contrast, \name delegates intricate code analysis to untrusted \vcgen[plural], minimizing the TCB.
}

\ignore{
\subsubsection*{SFI Enforcement and WebAssembly Protections} 
\hongbo{May need to move or remove this part}\danfeng{we can remove it} Previous work realized SFI in various systems~\cite{sfi:principles}, including in the context of confidential computing~\cite{sgx:deflection, sgx:occlum}.
WebAssembly (WASM) has gained significant attention in recent years.
Designed with memory safety and control-flow integrity~\cite {wasm:security}, WASM's protective features can fail when translating its code into machine-readable formats. 
For instance, Cranelift, the compiler utilized by Lucet and wasmtime, is known to generate binaries with sandbox escape vulnerabilities~\cite{cve:2021-32629, cve:2023-26489}.
A formally verified sandboxing compiler for WASM has been implemented to strengthen WASM binary compilation~\cite{wasm:provably-safe}.
Moreover, several verification methodologies for SFI-compliance in WASM are also proposed~\cite{wasm:zerocost, veriwasm}.
\name's \vcgen for a WASM SFI policy is built on top of their artifact.
While verifiers specifically targeting SFI exist~\cite{nacl-sandbox, rocksalt, armor,veriwasm}, they all assume a closed verification ecosystem, and their adaptability to diverse policies can be constrained, a limitation that \name addresses distinctively.
}

\subsubsection*{Smart Contracts and Bug Bounty Programs}
The adoption of bug bounty programs spans across sectors, fostering a collaborative environment where bug bounty hunters engage in identifying and reporting bugs for monetary rewards.
Research has delved into bug bounty policy designs to augment participation, exploring the viewpoints of both bounty platforms and hunters~\cite{finifter2013empirical,mingyi15empirical,votipka2018hackers,akgul2023bug}.
Recognizing the advantages of automation and transparency of smart contracts, researchers also pursued the integration of bug bounty platforms with smart contracts~\cite{tramer2017sealed, breidenbach2018enter}.
However, existing designs accept only submissions of bugs while ignoring reports of ``no-bug''.
Thus, they cannot assure verification results for binaries, which is crucial for verifiers.

\section{System Overview}
\label{sec:workflow}

\ignore{
\subsection{Design Guidelines}
\label{sec:philosophy}
\danfeng{given the space constraint, I'm not sure if we should keep this part. The points are highlighted in the introduction already.}

\name is predicated on the principles of ``open more'' and ``trust less'' while maintaining compatibility with versatile policies and providing assurance of verification results.
These guidelines (\textbf{G$\mathbf{i}$}'s) underpin our system's design:

\begin{itemize}
    \item [\textbf{G1}] \textit{Compact TCB}.
    \name minimizes the TCB by delegating intricate program analyses and constraint solving to untrusted entities.
    Trusted yet lightweight validators in the system render malicious input ineffective.
    This streamlined approach eases both manual inspection and the potential for future formal verification.

    \item [\textbf{G2}] \textit{Elimination of Dictatorial Parties}.
    To preclude the possibility of manipulated verification results, \name integrates TEE and blockchain technologies.
    This combination facilitates a decentralized verification service, obviating reliance on any single dictatorial entity.


    \item [\textbf{G3}] \textit{Support for Versatile Policies}.
    \name is equipped with a flexible and expressive \vc language to represent arbitrary program execution semantics.
    We model three distinct policies to demonstrate its versatility.

    \item [\textbf{G4}] \textit{Open to Untrusted Contributions}.
    \name embraces contributions from untrusted participants.
    The system's front-end \vc[s] are sourced from untrusted analyzers, which incorporates diverse binary analysis technologies.
    Conversely, back-end constraint solving results are crowdsourced from BBHs, encouraging them to develop innovative algorithms.

    \item [\textbf{G5}] \textit{Auditable Verification Pipeline}.
    Ensuring transparency and reproducibility, \name leverages TEE for remote attestation and blockchain for immutable record-keeping.
    This allows public auditing and enables users to re-verify results to confirm policy compliance via the published evidence.

\end{itemize}

\ignore{
\textit{\textbf{Trustworthy}: \name bears a compact TCB}.
By delegating intricate program analyses, disassembling, and constraint solving to untrusted entities, \name significantly narrows its TCB, easing future formal verification.
For instance, drawing inspiration from PCC~\cite{pcc-original}, \name separates \emph{untrusted} proof provider (employing sophisticated program analyses to formulate proofs) from \emph{trusted} proof validator, which assesses the correctness of proofs.
Besides, rather than trusting complex SMT solvers, we offload constraint solving to untrusted BBHs.

\textit{\textbf{Versatile}: various policies can be expressed and verified}.
\name supports a flexible and expressive proof language (an extension of the first-order predicate logic), harnessing constraint solving for enhanced adaptability to check proofs.
As evidence of its versatility, we model an SFI policy-checking paradigm for WASM and extend the language to incorporate a side-channel mitigation policy.

\changes{
\textit{\textbf{IP Protected}: the intellectual properties in proprietary code can be protected}.
Since \name only requires proof related to the policy, most semantics information irrelevant to the policy is truncated (above 80\% for WASM SFI).
Moreover, when TEEs are enabled in crowdsourced verification, no semantics information is revealed in the pipeline.}

\textit{\textbf{Open}: \name champions an open verification ecosystem}.
We permit active engagement across both the initiation and culmination phases of the workflow.
Front-end proofs are sourced from untrusted generators, who can incorporate diverse binary analysis technologies as long as the proofs remain valid.
Conversely, back-end constraint-solving results are crowdsourced from a myriad of BBHs, unrestricted in their methodological approaches.
BBHs might leverage established SMT solvers, devise proprietary algorithms, or even undertake manual resolutions.

\textit{\textbf{Auditable}:  transparency and continuous auditability underpin the verification process}.
Every component of \name is designed for public scrutiny.
TEE furnishes an execution environment conducive to remote attestation, reassuring participants of the service's authenticity.
Further, by synergizing with blockchain, \name ensures verification transparency for all binary users.
This approach not only enables public examination of the verification outcomes but also encourages potential future submissions to challenge and invalidate formerly certified code.
\hongbo{End of dropping.}
}

}




\begin{figure}[tb]
    \centering
    \includegraphics[width=0.7\textwidth]{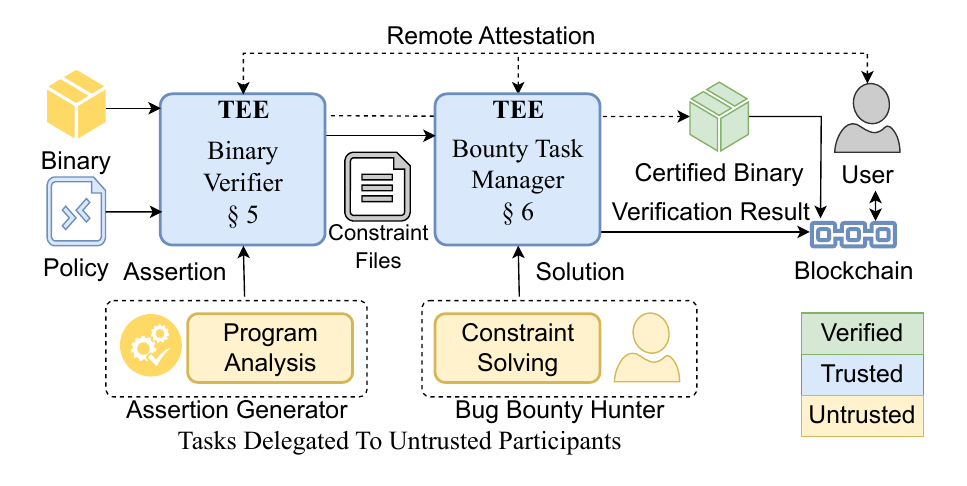}
    \Description{Description placeholder}
    \caption{System architecture of \name.}
    \label{fig:arch}
\end{figure}


\subsubsection*{High-level Workflow}
\name integrates two core components: a \textit{Binary Verifier} (BV) and a \textit{Bounty Task Manager} (BTM), offloading program analysis and constraint solving, respectively.
The workflow illustrated in \autoref{fig:arch} commences with the input binary and \vc.
The \vc is synthesized by an untrusted \vcgen conducting program analysis according to the selected policy.
The BV then validates the \vc[s] for correctness and formulates verification constraints based on the policy specification.
The constraint sets are output to files, which can be verified by constraint solvers.
However, as \name excludes SMT solvers from the TCB, the BV dispatches the constraint files to the BTM for further inspection.
The BTM maintains a bug bounty protocol, accepting submissions of solving results for the constraint sets from BBHs.
If a BBH submits a bug to the BTM, the BTM marks the binary as insecure after validation.
Finally, \name documents the verification results with evidence and publishes the certified binary on the blockchain, as the crowdsourced solutions are received via the bug bounty protocol. 
Users can thus obtain the certified program from the blockchain. \name also enables users to fetch the results and even reproduce the entire verification process by accessing the evidence on the blockchain.

\ignore{
As depicted in \autoref{fig:arch},
the workflow commences with a binary declaring security policies, and the proof of each policy is prepared by an untrusted proof generator.
After disassembling the binary and parsing the proof, the BV first preprocesses (\ding{172}) all inputs to proper internal formats.
Then, the proof validator scrutinizes the correctness of each proof (\ding{173}), and the policy checker formulates verification conditions that entail policy compliance (\ding{174}).
Verification conditions are further formed as SMT2 constraints, which are solved by untrusted reference solvers.
The BV early-rejects the binary if satisfiable.
Otherwise, it dispatched the constraint files to the BTM.
\changes{depending on the vendor's requirement on code confidentiality.}
The BTM further inspects the constraints and maintains a bug bounty protocol, accepting submissions of verification results from BBHs.
The BTM finally updates the verification outcomes after validation. The BV rebuilds and signs a binary from the verified disassembled code.
End users can query the verification service for the outcome.
}

\subsubsection*{Threat Model}
\autoref{fig:arch} highlights the verified, trusted, and untrusted entities.
We predominantly regard the BV and the BTM as trusted, shielded within hardware TEEs~\cite{intel:sgx-brief, amd:sev-snp, intel:tdx}.
However, trust does not extend to cloud providers or system administrators.
Our assumptions align with standards on blockchain and smart contracts: the smart contracts remain invulnerable, and the foundational blockchain is robustly secure~\cite{dembo2020everything}.
We require participation from a sufficient number of BBHs in the bug bounty program, assuming the presence of honest BBHs who comply with protocols and duly report any discerned vulnerability.
We also assume honest BBHs are capable of finding bugs.


Potential threats encompass malicious input from the binary, the \vc generation, and BBHs.
As the users simply make queries and conduct remote attestation, they cannot affect the verification results and are excluded from the discussion.
We acknowledge the capabilities of software vendors, who could potentially manipulate input binaries, such as embedding vulnerabilities and altering machine code.
Similarly, \vcgen[plural] may yield flawed or even meticulously crafted malicious \vc[s].
The BV handles these threats.
As the certified binary is signed by \name and the blockchain storage is append-only, the software vendor cannot modify it without re-verification.
Besides, BBHs may conceal bugs or even submit fake ones, which is also ruled out by
our protocol (see \autoref{sec:btm}).
Although we consider denial of service (DoS), \name does not cover side-channel attacks such as those in~\cite{xiong2021survey, wang2017leaky, li2021cipherleaks}.

\section{Binary Verification}
\label{sec:proof-checking}


\ignore{
\begin{figure}[tb]
    \centering
    \includegraphics[width=0.99\columnwidth]{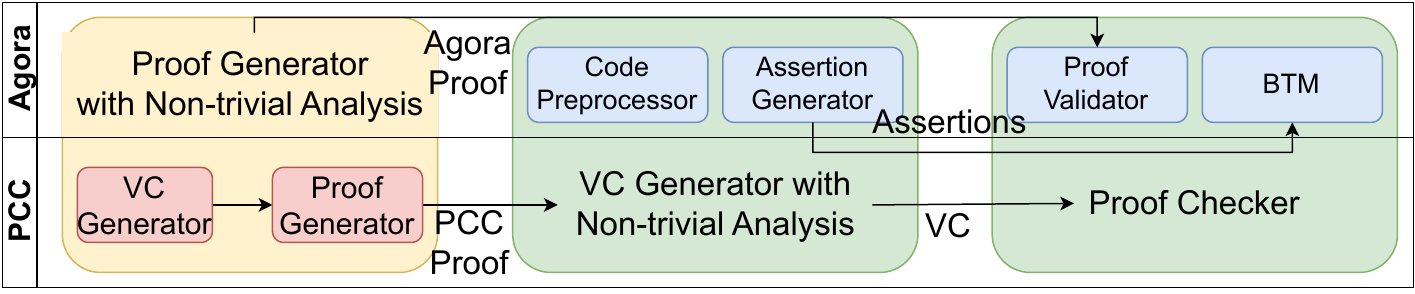}
    \caption{Comparison between the components of Agora and PCC}
    \label{fig:agorapcc}
\end{figure}
}


Creating an open and trustworthy verification framework presents several challenges. 
We highlight several innovations in \name's binary verifier (\autoref{fig:bv}): 

\ignore{
The proof-carrying code methodology effectively reduces the trust required in critical verification framework components, thereby shrinking the TCB size.
However, the consumer-side proof checkers still need to run nontrivial code analysis, such as Floyd-style verification condition generation~\cite{floyd1993}. 
For example, in the Cedilla Systems implementation of PCC~\cite{colby2000certifying}, the proof checker contains 23K lines of C code, without counting the size of the backend SMT solvers (still not fulfilling \textbf{G1}).
Moreover, PCC's applications, while methodically versatile, are typically tied to a specific policies, hindering the verification of versatile policies in one integrated system (violating \textbf{G3}). 
Besides, the choice of verification techniques is rigorously limited (violating \textbf{G4}).
For instance, both the code producer and consumer need to compute Floyd-style verification conditions to verify safe assembly extensions of ML programs~\cite{pcc-original}.
Replacing the proof generation methods on the producer side, with say symbolic execution, fails to validate the proof. Also, the proof in PCC programs are written into the binary, limiting proof generation to only the code producer (another violation of \textbf{G4}). 
}

\ignore{
Verification frameworks built on standardized interfaces, on the other hand, provide a more open solution, facilitating systematic communications between arbitrary verification components.
This enables the support for diverse policies (satisfying \textbf{G3}) and verification techniques.
However, such frameworks usually necessitate trusting all components involved in the verification pipeline, conflicting with the targets of a small TCB size (violating \textbf{G1}).
Also, there is no channel of incorporating contributions from \emph{untrusted} participants in providing verification proof (violating \textbf{G4}).}

\ignore{
As the advantages and shortcomings of each technique of the solutions above are clearly laid out, it's natural to consider picking up the good sides of the two methods by combining them together, i.e., using PCC's workflow while supporting versatile policies via standardized interfaces.
However, simple integration does not work.
On the one hand, previous modularized frameworks regard all components as trusted (e.g., a trusted proof generator generates proofs and send it into a trusted theorem prover without having to validate the proof), so it's unclear whether those standardized interfaces can be reused in the context of PCC, as the proof received via this interface is not trusted and must be rigorously validated. 
On the other hand, even if such a proof checker can be realized for the standardized interface, the TCB size paid for adopting this interface may still be very large.
This situation deteriorates as versatile policies require support for more expressiveness.
Therefore, naively trying to integrate them leaves us with a ``robbing Peter to pay Paul'' dilemma. \quan{Hongbo and I think this paragraph is not well-written, because the shortcomings of combining the two methods seems dodgy.}
}



\ignore{
\name is inspired by Proof-Carrying Code (PCC)~\cite{pcc-original}, wherein the source code or binary code is untrusted, and its safety is established by a proof alongside the code. In PCC, a code provider publishes code with its safety proof, and a code consumer uses the proof and runs verification locally to establish trust in the associated code. Notably, PCC removes trust in both code and proof, and it reduces the cost of verification as checking a proof is usually much simpler than generating a proof (e.g., loop invariants). 
However, PCC falls short of building an open verification system of security properties for a few reasons:
\quan{Address some reviewer comments here. Particularly, I want to address: why is \name different from previous pipeline workflows, e.g., Frama-C, ETB, etc. \textit{Standardized interfacing} seems to be a major standpoint of (several) previous tools. Our work does not advance more in that direction. If thinking carefully, we are basically \textit{PCC with a standardized interface to support various policies} - on the verification front-end.}
\hongbo{new storyline: PCC is not versatile; standardized interfacing can handle various policies; their TCBs are still way too large (although PCC is better); \name achieves versatility with smaller TCB size, benefiting from delegation and open.}

\hongbo{NOTE the bullets are substituted by the following paragraph.}

\begin{itemize}
    \setlength\itemsep{0em}
    \item While the general PCC idea~\cite{pcc-original} can be applied to a variety of policies, its instantiation is usually tied to a specific policy, or even to a specific verification method.
    For instance, both the code producer and consumer need to compute Floyd-style verification conditions~\cite{floyd1993} to verify safe assembly extensions of ML programs~\cite{pcc-original}. Replacing the verification methods on the consumer side, with say symbolic execution, fails to validate the proof.

    \item \hongbo{Openness: PCC couples proof-gen and verifier} For each policy, the \vcgen is coupled with a corresponding verifier, where the correspondence is built on a series of typing rules \quan{not necessarily} \xchanges{a series of predefined protocols, e.g., typing rules.?}.
    However, \xchanges{the protocols vary largely between different instantiations.}

    \item PCC fully trusts proof validators running on the consumer end. Usually, proof validators still need to run nontrivial code analysis, such as Floyd-style verification conditions, resulting in a large TCB (e.g., 23K lines of C in the Cedilla Systems implementation of PCC~\cite{colby2000certifying}, excluding third-party libraries such as SMT solvers). \hongbo{May mention the constraint solving part is not trusted in our framework?}

    \item While prior work \cite{necula1996safe, pcc-original,colby2000certifying, tal, pcc:network, pcc:authentication, pcc:authorization} has verified various safety policies under the PCC approach, it is unclear if PCC is applicable to nontrivial security policies, such as SFI~\cite{sfi:principles}.
\end{itemize}
}

\ignore{

\xchanges{However, a naive combination of these two techniques falls short in meeting all goals in \autoref{sec:philosophy}. 
While the general PCC methodology~\cite{pcc-original} can be applied to a variety of policies, its instantiation is usually tied to a specific policy, or even to a specific verification method, hindering the versatility of policies (\textbf{G3}).
For instance, both the code producer and consumer need to compute Floyd-style verification conditions~\cite{floyd1993} to verify safe assembly extensions of ML programs~\cite{pcc-original}. Replacing the verification methods on the consumer side, with say symbolic execution, fails to validate the proof. For each policy, the proof generator is coupled with a corresponding verifier, where the correspondence is built on a series of predefined logical formulae. 
To address this problem, prior works~\cite{cruanes2013etb,kirchner2015framac} have proposed standardized interfaces to make systematic communication between arbitrary verification components possible, enabling the support for different policies and verification methods \hongbo{do these methods meet O2?}.
However, such frameworks usually require trust in all components involved in the verification pipeline, conflicting with a small TCB size (\textbf{G1}).
Nevertheless, proof validators running on the consumer end are fully trusted in PCC, extending the TCB size.
Usually, proof validators still need to run nontrivial code analysis, such as Floyd-style verification condition generation, resulting in a large TCB (e.g., 23K lines of C in the Cedilla Systems implementation of PCC~\cite{colby2000certifying}, without counting the size of the backend SMT solvers).}

Proof generation separates binary analysis from the verification process, thereby facilitating flexible proof generation and supporting a variety of policies~\cite{pcc-original}.
This adaptability implies that security properties can be proven using diverse program analysis methodologies customized for specific policies.
While generating such proof is usually intricate, validating that the proof is consistent with the given binary is considerably simpler.
This enables the offloading proof generation outside the TCB, as the BV can validate the proof.
\danfeng{The following still sounds straightforward. What are the real challenges?}
However, the original PCC proves the safety policy based on a set of logic/type inference rules.
Extending these policies can be strenuous, as new policies necessitate new rule sets.
We also want to avoid such endless patching to the rules.
By leveraging constraint solving as the backend, \name supports flexible proof generation and a variety of policies.
Thus, the BV needs to operate on a carefully crafted proof language to back this versatility.
}


\begin{itemize}
    \item \name supports untrusted \emph{\ndsschanges{\vcgen}[s]} 
    that emit \ndsschanges{\vc[s]} written in its \vc language (\autoref{sec:grammar}). The assertion language is versatile (i.e., agnostic of verification technique and security policy), striking a balance between expressiveness and simplicity: it can encode various security policies (\autoref{sec:casestudy}) while still allowing  
    a lightweight validator (\autoref{sec:proof-processing}) to check the correctness of those untrusted \vc[s].
    \item \name removes the trust in complex theorem provers.
    The \vc[s], together with the proof \regulation[s] generated for policy compliance check (\autoref{sec:policy}), form a verification constraint set,
    which is directed into the BTM (\autoref{sec:btm}) for further verification.
\end{itemize}

\begin{figure}[tb]
    \centering
    \includegraphics[width=0.7\textwidth]{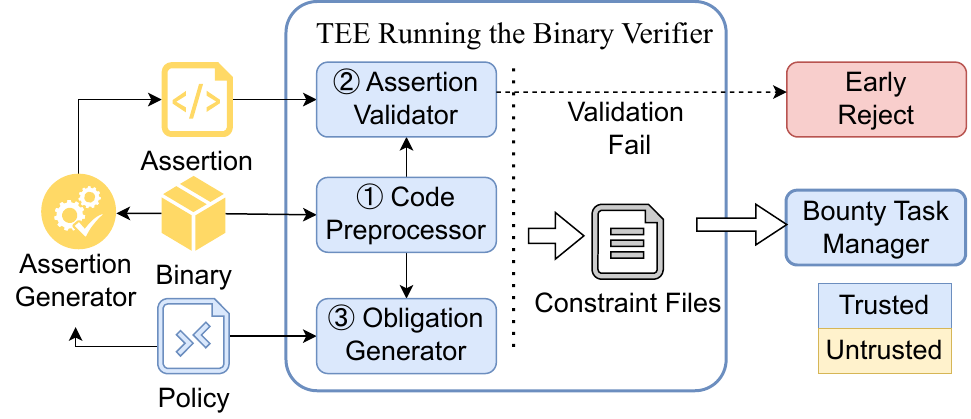}
    \Description{Description placeholder}
    \caption{Overview of the Binary Verifier's workflow.
    }
    \vspace{-2ex}
    \label{fig:bv}
\end{figure}

\ignore{
\subsubsection*{Comparison with PCC}
While Agora is inspired by PCC to remove trust on code producers, it achieves openness and trustworthiness that are absent in PCC:}
\ignore{
\begin{itemize}
    \item By removing trust on VCGen, Agora substantially reduces TCB size.
    \item Through its VC language, Agora decouples VCGen from a specific policy or program analysis technique. While PCC has the same VCGen deployed on both code consumer and producer, Agora incorporates \textit{any} VCGen generating Agora VC, promoting openness.
    \item By outsourcing constraint solving, Agora also promotes openness and trustworthiness on the verification backend.
\end{itemize}
}

\ignore{One insight of our design is that, unlike PCC, \name relieves the use of type systems which strongly couples with heavyweight compilation toolchains. Instead, the basic building blocks of \name, namely operational semantics and safety predicates, are expressed in first-order logic, which makes it generic to support a large class of policies. For simplicity, our prototype does not model self-modifying code.
Our current modeling of x86-64 ISA, although partially implemented, is to enough verify the two policies discussed in the paper. Extending it with more comprehensive semantics is rather straightforward.}

\subsection{Policy Specification and Proof Obligation}
\label{sec:policy}

Following convention, a policy specification in \name is formalized as a set of \textit{security conditions} and \textit{axioms} on abstract or concrete states of a program.
The security condition set represents the requirements for certain instructions or code patterns. 
For example, for control-flow restriction in SFI, one of the security conditions would specify that all $\tt jmp$ instructions must target valid code locations.

For each instruction or code pattern whose security condition (denoted as a Boolean formula $\mathcal{S}$) is not empty, we denote the concrete or abstract program states at that point as $\mathcal{F}$. There are two kinds of security policies regarding how to generate \emph{proof obligations} (i.e., a logical statement that must hold to ensure policy compliance): \emph{may} and \emph{must} policies\footnote{Inspired by may or must data-flow analysis in compiler optimization.}. 
In a \emph{may} policy, proof obligation is a disjunction of $\mathcal{S}$ and $\mathcal{F}$ (e.g., $\mathcal{F}\rightarrow \mathcal{S}$, which is the same as $\neg \mathcal{F} \lor \mathcal{S}$). Hence, removing program state $\mathcal{F}$ from proof obligation is \emph{sound} in this case. In a \emph{must} policy, however, proof obligation is a conjunction of $\mathcal{S}$ and $\mathcal{F}$ (e.g., $\mathcal{F}\land \mathcal{S}$). Hence, no program state in $\mathcal{F}$ should be removed from proof obligation in order to guarantee soundness. We discuss how may and must policies affect assertion validation in \autoref{sec:proof-processing}, and show SFI as a may policy in \autoref{sec:example-sfi} and IFC as a must policy in \autoref{sec:example-ifc}. 



On the other hand, axioms are predicates assumed to always be true, which can represent properties and relations of a specific binary or the runtime environment (see~\autoref{sec:sfi-cfi}). 

\subsection{Code Preprocessing}
\label{sec:ssa}
Similar to other binary analysis tools~\cite{angr, song2008bitblaze}, we use an intermediate representation (IR) to encode the semantics of instructions. 
Notably, all machine state changes (e.g., set/clear flags and stack/instruction pointer adjustments) are explicitly written into the IR. 
For example, the \emph{Code Preprocessor} takes the code $\tt 0x1000: cmp~\%rsi,\%rax$ and
 produces 
\[{\tt rip} = \addr{1003}, {\tt cf} = ({\tt rax} < {\tt rsi}), {\tt zf} = ({\tt rax} = {\tt rsi})...\]
where we omit a few other modified flags for simplicity.

To aid subsequent analysis steps, the Code Preprocessor performs basic static single-assignment (SSA) transformation and control-flow graph (CFG) construction on the IR.
For example, the SSA form of the above example can be
\begin{gather*}
{\tt rip.1} = \addr{1003}, {\tt cf.1} = ({\tt rax.1} < {\tt rsi.0}), 
{\tt zf.1} = ({\tt rax.1} = {\tt rsi.0})...
\end{gather*}
\subsection{Assertion Language}
\label{sec:grammar}

\name's assertion language is defined in \autoref{fig:proof}.
This language should be versatile to support various policies and verification techniques. Meanwhile, it should be well-structured and \emph{simple} enough to ease \vc validation.

In a nutshell, a line of \vc is associated with one instruction in the assembly code, stating a policy-specific fact (written as a Boolean expression on machine state) that holds after the execution of the instruction.
Boolean expressions ($\bnf{bool\_exp}$)
and arithmetic expressions ($\bnf{arith\_exp}$) are mostly standard: they include relations and computations on constants ($\bnf{num}$), general-purpose registers ($\bnf{reg}$) and flags ($\bnf{flag}$).
Moreover, for expressiveness, the \vc language models the stack frame ($\bnf{stack}$) of the function being examined with $[({\tt rsp} \mid {\tt rbp})\pm \bnf{num}]$ and a size prefix ($\mathtt{q}$: qword, $\mathtt{d}$: dword, $\mathtt{w}$: word, $\mathtt{b}$: byte).
%
The $\tt ite$ operator is a conditional operation on both arithmetic and Boolean expressions.
Note that the assertion language intentionally omits the quantifiers on variables (e.g., registers), as the quantifiers are determined by the policy specification (\autoref{sec:casestudy}).

\begin{figure}
  \centering
  \setlength{\abovecaptionskip}{0.cm}
  \footnotesize
  \begin{align*}
    \bnf{\vc}           \; \rightarrow \; & \bnf{num} : \; \bnf{bool\_exp}                                       \\
    \bnf{bool\_exp}     \; \rightarrow \; & {\tt true} \choice {\tt false} \choice {\tt not} \; \bnf{bool\_exp} \choice \bnf{flag} \choice  \\
                                          & \bnf{bool\_exp} \; \bnf{bool\_binop}  \; \bnf{bool\_exp}  \choice                               \\
                                          & \bnf{arith\_exp} \; \bnf{arith\_cmpop} \; \bnf{arith\_exp} \choice                              \\
                                          & {\tt ite} \; \bnf{bool\_exp} \; \bnf{bool\_exp} \; \bnf{bool\_exp} \choice                      \\
                                          & \bnf{predicate} \; ( \bnf{arith\_exp} ) +                                                       \\
    \bnf{arith\_exp}    \; \rightarrow \; & \bnf{num} \choice \bnf{reg} \choice \bnf{stack} \choice \bnf{symbol} \choice                    \\
                                          & \bnf{arith\_exp} \; \bnf{arith\_binop} \; \bnf{arith\_exp} \choice                              \\
                                          & {\tt ite} \; \bnf{bool\_exp} \; \bnf{arith\_exp} \;  \bnf{arith\_exp}                           \\
    \bnf{var}           \; \rightarrow \; & \bnf{reg} \choice \bnf{flag} \choice \bnf{stack} \choice \bnf{seg} \choice \bnf{symbol}         \\
    \bnf{num}           \; \rightarrow \; & {\tt 0} \choice {\tt 1} \choice {\tt 2} \choice {\tt 3} \choice ...                             \\
    \bnf{reg}           \; \rightarrow \; & {\tt rsp} \choice {\tt rbp} \choice {\tt rax} \choice {\tt rbx} \choice {\tt rcx} \choice ...   \\
    \bnf{flag}          \; \rightarrow \; & {\tt cf}  \choice {\tt zf}  \choice {\tt of}  \choice {\tt sf}  \choice ...                     \\
    \bnf{stack}         \; \rightarrow \; & ({\tt q} \choice {\tt d} \choice {\tt w} \choice {\tt b} ) 
                                                                                   \;  \; ({\tt rsp} \choice {\tt rbp}) [ \pm \bnf{num} ]   \\
    \bnf{symbol}        \; \rightarrow \; & {\tt GB} \choice {\tt GT} \choice {\tt GTS} \choice {\tt GTSAddr} \choice ...                   \\
    \bnf{predicate}     \; \rightarrow \; & {\tt FnPtr} \choice {\tt JmpOff} \choice {\tt JmpTgt} \choice ...                               \\
    \bnf{arith\_binop}  \; \rightarrow \; & + \choice - \choice \times \choice \ll \choice \gg \choice \& \choice ...                       \\
    \bnf{arith\_cmpop}  \; \rightarrow \; & > \choice \geq \choice < \choice \leq \choice =                                                 \\
    \bnf{bool\_binop}   \; \rightarrow \; & \land \choice \lor \choice \to                                                          \\
  \end{align*}
  \vspace{-5ex}
  \Description{Description placeholder}
  \caption{\name's \vc grammar.
  }
  \vspace{-2ex}
  \label{fig:proof}
\end{figure}




A more interesting feature of the \vc language is its support for policy-related symbols ($\bnf{symbol}$) and predicates ($\bnf{predicate}$), to express abstract properties on program state.
For example, each binary compiled by the Lucet compiler stores function pointer entries with 16-byte alignment in the binary’s data section (table $\tt guest\_table\_0$), and we can check that each pointer stored in the table is a valid function pointer. In the assertions, we 
introduce a special symbol $\tt GT$ and facts on the symbol, such as (${\tt GT} = \addr{2000}$) where \addr{2000} is the table location in the analyzed binary. 
Moreover, the \vcgen introduces a predicate $\tt FnPtr$ and a derivation rule that each loaded value $v$ from ${\tt base + offset}$, where 
$\tt base = GT$ and offset obeys a specific pattern, makes $v$
a valid function pointer (i.e., ${\tt FnPtr}~v$). We elaborate the details in~\autoref{sec:sfi-cfi}.


\ignore{
\quan{This paragraph to be removed due to the excessive discussion on PCC.}
\hongbo{TCB components, proof size, proof simplicity, proof generality, decoupling.}
\ndsschanges{
While \name's design guidelines align PCC to remove trust on code producers, it further decouples the proof language from type inference rules either from VCGen in conventional PCC or from typed assembly language in FPCC.
Such design benefits \name in several aspects.
Through \name's \vc grammar, assertions required from untrusted analyzers are all semantic, which can be achieved with modest modifications on current analyzers.
Therefore, \name is more adaptive to existing analyzers, opening to contributions from the verification community.
}
}

\ndsschanges{Note that due to the goal of ``trust less'', we have to balance expressiveness and simplicity of the \vc language.
%
As a result, \name does not support policies involving fine-grained heap memory reasoning and nested quantifiers, for example. Yet, we found the \vc language is powerful enough for verification of nontrivial policies (e.g., SFI and IFC) without bloating the TCB.}

\subsection{Assertion Validation}
\label{sec:proof-processing}

\changes{Explain quantifiers. Local validation $\simeq$ quantifier-free; global validation $\simeq$ quantified}
\vc[initcap] validation should be simple and lightweight to offload the complexity in binary analysis (i.e., \vc generation).
For each \vc line in \autoref{fig:proof}, the \emph{\vc[initcap] Validator} 
first performs an SSA transformation based on the SSA IR provided by the Code Preprocessor.
For the example in $\tt 0x1000: cmp~\%rsi,\%rax$ (see \autoref{sec:ssa}), where  ${\tt cf}$, ${\tt rax}$, and ${\tt rsi}$ are transformed into  ${\tt cf.1}$, ${\tt rax.1}$, and ${\tt rsi.0}$, respectively; an \vc consists of ${\tt cf} = ({\tt rax} < {\tt rsi})$ is translated into ${\tt cf.1} = ({\tt rax.1} < {\tt rsi.0})$. For simplicity, all IR and \vc[s] hereon are assumed in their SSA forms.

For both may and must policies defined in \autoref{sec:policy}, the \vc[initcap] Validator \emph{validates} the correctness
of the \vc in a two-stage (instruction- and function-level) fashion.
\if false
Depending on the validation outcome, the BV chooses one of the following options: \quan{These itemize are no longer needed since they are rewritten below.}
\begin{itemize}
    \setlength\itemsep{0em}
    \item reject the \vc early if a line is locally invalid,
    \item rewrite the locally validated \vc into constraint \textit{facts} and send it to the constraint checking environment,
    \item rewrite \vc[s] which cannot be locally validated into constraint \textit{checks} and request the solver to check them.
\end{itemize}
\fi
%
\if false
------ \emph{rewriting the following part} ------
We first check if any contradiction exists (from a satisfiability standpoint) between the IR and proofs:
\[
    \mathcal{IR}_{SSA}[a] \land \mathcal{R}_{SSA}[a] \rightarrow \bot
\]
\danfeng{$\bot$ means false? how are the free variables quantified here? Logically, this is the same as the negation of $(\land \mathcal{IR}_{SSA}[a]) \land (\land \mathcal{R}_{SSA}[a])$. what's the meaning of the formula?}
If a contradiction can be identified, the proof is determined to be invalid\danfeng{how can you early reject it with local semantics?} and the verification process terminates, returning a \textit{early} rejection.
Otherwise, Proof Validator checks if the IR can logically imply each of the relations in $R_{SSA}[a]$:
\[
  \mathcal{IR}_{SSA}[a] \rightarrow r, \quad r \in \mathcal{R}_{SSA}[a]
\]
Any relation $r$ that satisfies this implication is elevated to be a \textit{fact}, and thus added to the constraint environment, which includes a set of validated relations rewritten as constraints\done\danfeng{what is the constraint environment and why do we need it? we need to introduce that earlier in the paper}. 
If there exists some $r$ such that it can't be implied by $\mathcal{IR}_{SSA}[a]$, they must be checked at the constraint environment later via constraint solving. 
For example, for the two proofs at \addr{6274}, it's only possible to imply the first relation from the IR:
\begin{align*}
    {\tt rip}.6 = \addr{6289} \; \land \; &{\tt rdi}.2 = {\tt rdi}.1 \ll 4 \\
                    \rightarrow   \; &{\tt rdi}.2 = {\tt rdi}.1 \ll 4
\end{align*}
However, the second relation ${\tt rsi}.3 > \addr{1000}$, although being true due to the execution of the previous instruction, cannot be \textit{locally} validated through the above method. 
This is simply because the Proof Validator is agnostic of the state derived from the semantics irrelevant to the instruction currently being analyzed. 
Since we can't validate its correctness or reject it due to a contradiction, some further checks are necessary.
%
------ \emph{rewritten part (old version) end} ------
\fi
%
%
More precisely, we denote the IR semantics at code address $i$ as a set of propositions $\mathcal{IR}[i]$,
and the \vc[s] as another set $\mathcal{A}[i]$.
Besides, we denote a \emph{fact} set $\mathcal{F}$, initialized to be empty, to maintain all known facts being analyzed so far, and denote $\mathcal{PC}[i]$ for the path condition. Then for each $a \in \mathcal{A}[i]$:
\begin{itemize}
  \item The validator first checks if $a \in \mathcal{IR}[i]$. If so, we say $a$ is validated \emph{locally} (i.e., at the instruction level). The validated $a$ is added to $\mathcal{F}$ afterward if the policy is a \emph{may} policy. Otherwise, $\mathcal{IR}[i]$ is added to $\mathcal{F}$.
  
  \item If $a \notin \mathcal{IR}[i]$, the validator generates a function-level check of $a$ against $\mathcal{F}$ (the concrete check is policy specific). A successful check adds $a$ to $\mathcal{F}$ if the policy is a \emph{may} policy, or adds both $a$ and $\mathcal{IR}[i]$ to $\mathcal{F}$ if it is a $\emph{must}$ policy. A failed check terminates verification right away. 

  \ignore{
  \quan{I added $\mathcal{PC}$ in the formula. Seems like if $\mathcal{PC}$ is added, there's no need to make the awkward discussion anymore. Also, I think we should say that the function-level validation is done in the BTM (solvers), which goes along with the next paragraph. And that shows why the LOC is only 1k for a validator. }
  \hongbo{strange logic} \quan{That is related to the $P \rightarrow Q$ to $P \land \neg Q$ transition. I agree maybe it's too low-level, and we don't need it here}
  }
\end{itemize}


Next, we use a may policy SFI to elaborate the process. Revisiting the instruction $\tt 0x1000: cmp~\%rsi,\%rax$ with
\begin{align*}
  \mathcal{IR}[\addr{1000}] = &~ \{{\tt rip.2} = \addr{1003}, 
        {\tt cf.1} = ({\tt rax.1} < {\tt rsi.0}), {\tt zf.1} = ({\tt rax.1} = {\tt rsi.0}), ... \}
\end{align*}
Assume that before the validation, $\mathcal{F} = \{\tt rbx.1 = 1\}$ and $\mathcal{PC}[\addr{1000}] = {\tt true}$ (i.e., no branch before \addr{1000}). Consider the validation of different untrusted \vc[s]:
\begin{itemize}
    \item $ {\tt cf.1} = ({\tt rax.1} < {\tt rsi.0}) $:
    The \vc can be validated at the instruction level as it can be found in $\mathcal{IR}[\addr{1000}]$. Moreover, the \vc is added to $\mathcal{F}$ after validation.
    \item ${\tt cf.1} = ({\tt rax.1} > {\tt rsi.0})$:
    Since no match can be found in $\mathcal{IR}[\addr{1000}]$, the validator defer it to the function-level validation, checking $\mathcal{F} \land \mathcal{IR}[\addr{1000}]\rightarrow ({\tt cf.1} = ({\tt rax.1} > {\tt rsi.0}))$. Since this (universally-quantified) check 
    fails (e.g., when ${\tt rax.1}=1,{\tt rsi.0}=2$), the \vc is rejected when the whole function is checked. 
    \item ${\tt rbx.1} > {\tt 0}$:
    Similarly, the \vc is deferred to the function-level validation. However, since $\mathcal{F} \land \mathcal{IR}[\addr{1000}] \rightarrow ({\tt rbx.1} > {\tt 0})$ always holds, the \vc is validated and added to $\mathcal{F}$. 
\end{itemize}

Note that for a must policy, $\mathcal{IR}[i]$ is always added to $\mathcal{F}$ to guarantee soundness (definition in \autoref{sec:policy}). 
A must policy takes the conjunction of program facts ($\mathcal{F}$) and other conditions for security check. However, this is not required for a may policy, enabling an optimization for may policies. In practice, we found that most \vc[s] (excluding policy-specific predicates) can be validated at the instruction level. In our evaluation (\autoref{sec:evaluation}), the optimization on a may policy saves about 83\% in constraint size and accelerates the running time of the constraint solving step by 70\%.

\if false
To handle the last case, the BV does the following.
A function-level proposition $\mathcal{F}$ is constructed to maintain all known facts. 
$\mathcal{F}$ stores all locally validated relations in logical conjunction. 
Additionally, if for some address $a$, there is an $r \in \mathcal{R}_{SSA}[a]$ which cannot be locally validated, the whole $\mathcal{IR}_{SSA}[a]$ is added to $\mathcal{F}$\done\danfeng{does not sound very efficient}.
We denote another conjunction $\mathcal{C}$, to store all outstanding relations\hongbo{is outstanding relations a term?}\quan{``outstanding'' is an adjective describing the relations, it's not a term ``outstanding relations''} to be checked (from the last case). 
Those that cannot be added into $\mathcal{F}$ in the previous step are added into $\mathcal{C}$.
After all instructions of a function are done with local validation, the complete set of $\mathcal{F}$ and $\mathcal{C}$ will be formed.

Next, the two propositions are taken to check for satisfiability of non-locally validated proof by constraint solving, as shown below\danfeng{again, explain why should we check this formula}. 
\[
    {\tt SAT}(\mathcal{F} \land \neg \mathcal{C}) = 
    \begin{cases}
        sat    & \implies reject \\
        unsat  & \implies accept
    \end{cases}
\]

\hongbo{maybe we can modify the notations. $\mathcal{C}$ represents the logical conjunction of a set of propositions $C$, and we may need to separate set from conjunction notations.}
If the solution is $sat$, $\mathcal{C}$ must have at least one relation clause which cannot be checked to match the semantics of the function or $\mathcal{C}$ itself is unsatisfiable, \hongbo{maybe we should check for each c in C? C can be in consistent}, resulting in a rejection. 
Otherwise, all relations in $\mathcal{C}$ are accepted to be valid, and \quan{we performed the evaluation on the following first two variants but not the third one, as it is hard to evaluate from the current implementation}:
\begin{itemize}
  \item $\mathcal{C}$ joins $\mathcal{F}$ to become the final set of facts, or $\mathcal{F'} = \mathcal{F} \land \mathcal{C}$ \\
        \done\danfeng{why keeping both F and C? seems like there is some redundancy: all non-local proofs are redundant. Why should we keep them?}.
  \item $\mathcal{C}$ is discarded and $\mathcal{F}$ is used as the final set of facts \quan{This method is 5-10\% better than the first one, discussed in eval section}.
  \item $\mathcal{C}$ is added to into $\mathcal{F}$, removing the minimal subset of $\mathcal{F}$ required to implies $\mathcal{C}$. \\
        $\mathcal{F}' = \mathcal{F} \land \mathcal{C} - min(\mathcal{F})$ \quan{mixture of set operation and logical operation, need to be fixed.}
        \hongbo{practically, it's hard to calculate $min(\mathcal{F})$, as only binary analysis knows what facts contribute to the proof but the BV does not know.}
\end{itemize}

\fi

\subsubsection*{Stack Modeling}
To verify certain policies, it is necessary to model the stack memory contents since compiled code sometimes stores a value into a stack memory cell and later loads it back to a register.
Compared to modeling arbitrary memory (e.g., on the heap), which significantly complicates \vc validation, 
modeling stack memory is much simpler: in most x86 code, 1) stack pointers are usually adjusted through constant addition/subtraction ($\tt rsp = rsp \pm \bnf{num}$), and 2) stack is usually accessed via ${\tt rsp}$/${\tt rbp}$ with constant offsets  (e.g., $\tt mov ~ rax, [ rsp + 8]$).
Whenever these criteria are met in the analyzed function, the BV models each memory cell on the stack (e.g., $[{\tt rsp + 8}]$) as a variable and applies standard SSA transformation when its content is modified.

\subsection{Satisfiability Check}
\label{sec:satcheck}
The BV gathers all of the aforementioned \vc validation and proof obligation constraints
and translates them into a \textit{constraint file} in SMT2 format. Depending on the policy (\autoref{sec:casestudy}), variables in constraints are either universally or existentially quantified. In either case, policy compliance boils down to a satisfiablity problem: either a \sat model implies security ($\exists$-quantified), or \unsat of the negation of the constraints implies security ($\forall$-quantified).

Ordinarily, program verifiers utilizing SMT solvers completely \emph{trust} them, which fails both openness and trustworthiness goals.
For openness, promoting one ``monopolistic'' solver prevents other solvers, including novel and experimental solvers from contributing to the verification pipeline.
Such solvers might be SMT competition winners, like bitwuzla in SMT-COMP 2023, FPArith track~\cite{niemetz2023bitwuzla, smt-comp2023:FPArith}.
For trustworthiness, mature SMT solvers are complex (e.g., Z3 has 532 KLOC) but have been discovered with soundness issues and vulnerabilities.
This is further deteriorated by the fact that solvers sometimes either time out or report \texttt{UNKNOWN} for complex constraints\footnote{In practice, it is common to use multiple solvers and hope that at least one reports \sat or \unsat. However, it still rules out other solving methods as it promotes a limited set of ``oligopolistic'' solvers.
Moreover, the TCB size is not significantly reduced in this approach.}.
To tackle the challenge, \name alters the traditional verification workflow by substituting SMT solvers with the BTM in TCB, 
which we elaborate in \autoref{sec:btm}.

\ignore{
Finally, the BV generates constraints according to the policy.
The policy is a set of rules that inserts assertions based on the assembly code, and multiple policies can be launched simultaneously.
These rules can be flow-sensitive and may rely on the hint.\danfeng{cannot follow the sentence}
Therefore, policy compliance is converted to checking a series of assertions, and the validated proof is regarded as the assumptions that presumably derive the assertions.
\name negates the assertions and expects an unsatisfiable result.
The BV creates a stack frame~\cite{smt26}\done\danfeng{what is stack frame?} for each assertion.
Such frame ensures that the checking does not affect the environment after the frame is popped off\danfeng{what does pollute mean here?}.
The assumptions and assertions, encoded as well-formed \texttt{smt2} constraint format, are forwarded to the BTM for further process\danfeng{we need to provide more details here. How exactly are the constraints generated?}.
\hongbo{example?}
We also use an SMT solver here to filter out policy-incompliant binaries to reduce the cost of running a bounty program.

\begin{itemize}
    \item Assertion \& task generation.
\end{itemize}
}
\section{Case Study}
\spchanges{To demonstrate \name's flexibility, we extract specifications of three distinct security policies and migrate their verifier onto \name.
Variants with relatively less complexity of these policies have also been enforced by previous work in CC~\cite{sgx:deflection, sgx:occlum, pobf, sgx:tsgx}.
Our case study shows that \name is suitable for the quickly evolving cloud CC use scenario.}
\label{sec:casestudy}

\subsection{Verifying SFI}
\label{sec:example-sfi}

We first use \name to verify SFI for x86 binaries compiled from WASM programs using Lucet compiler.
We extract the formally verified SFI specification and verification conditions from VeriWASM~\cite{veriwasm}, and verify the same in \name.
As Deflection~\cite{sgx:deflection} verifies a weaker SFI policy, we only elaborate on VeriWASM here and leave further discussions in \autoref{sec:tcb-analysis}.
In a nutshell, VeriWASM (and our method) verifies that a secure binary obeys the following restrictions:
1) any memory access (read or write) is within the bounds of a predefined sandbox, and
2) any control-flow-related instruction (especially an indirect call/jump or a return instruction) targets a valid code location in the program.
By the SFI policy specification, for some address $i$ where the security condition is $\mathcal{S}[i]$, 
its proof obligation has the form: 
$\forall \vec{v}.~(\mathcal{F} \land \mathcal{PC}[i]) \to \mathcal{S}[i]$. Hence, SFI is a may policy.



\subsubsection{Memory Access Safety}
\label{sec:sfi-mem}
One important aspect of SFI is to ensure that all memory accesses (including reads/writes to the heap, stack, and global memory) are within the corresponding regions in a sandbox. 
Each memory region is continuous but disjoint from others. 
Hence, the task is to check: for \emph{any} memory access at address ${\tt Addr}$, ${\tt Addr\in Heap \cup Stack \cup Global}$. We elaborate on each type of memory access next.

\subsubsection*{Heap Safety} Lucet defines the heap space as a continuous region with the size of 8GB, including a 4GB default usable memory space followed by a 4GB guard zone \cite{efficientsfi}.
Lucet compiles its code so that at the start of any function, the register $\tt rdi$ holds the address of the heap base, a fixed location in the runtime. 
Hence, for each $\tt call$ instruction where $\tt rdi.x$ is the current version of register $\tt rdi$, we validate the \regulation $\tt rdi.0 = \tt rdi.x$, which if successful, will establish an invariant on the heap base.

To verify that each memory access instruction (e.g., $\tt mov$) complies with heap safety, the \regulation[initcap] Generator generates proof obligations based on the following rule:
\begin{align*}
  \tt{HEAP\text{\ding{51}}} \equiv~ 
                    &{\tt rdi.0} < {\tt Addr} \land {\tt Addr} < {\tt rdi.0 + 8{\tt G}}
\end{align*}
where ${\tt Addr}$ is the target of the verified $\tt mov$ instruction, and ${\tt rdi.0}$ is the value of ${\tt rdi}$ at the beginning of the function.

\subsubsection*{Stack Safety}
\label{appx:stack}
The safety of stack access is largely the same as a heap access, with a few minor differences.
The base-stack pointer $\tt rbp$ at function start is an anchor point of its stack frame where there is an 8KB read-only section above it (holding the return address and spilled arguments of the function), as well as a 4KB read and write region below it (holding the function's local variables). 
The stack pointer $\tt rsp$ is modified via instructions such as $\tt pop$ and $\tt push$ in the function.

To verify memory access instructions comply with stack access safety, the following rules must be checked:
\begin{align*}
    \tt{STACKR\text{\ding{51}}} \equiv~ &\tt{rbp.0} - 4\tt{K} < \tt{Addr} \land \tt{Addr} < \tt{rbp.0} + 8\tt{K} \\
    \tt{STACKW\text{\ding{51}}}\equiv~ &\tt{rbp.0} - 4\tt{K} < \tt{Addr} \land \tt{Addr} < \tt{rbp.0}
\end{align*}
where ${\tt Addr}$ is the target of the validated memory access instruction.
Following the same workflow in verifying heap access safety, we can verify stack safety. 

\ifexample
\autoref{code:heap} shows a simplified code snippet of a heap write and corresponding \vc[s] generated by a proof-of-concept \vcgen. 
%
\begin{lstlisting}[label={code:heap}, escapechar=|, caption={Heap access example in SSA-like assembly with its assertions in colored background.}, float, linebackgroundcolor={
\ifnum\value{lstnumber}=3\color{purple!20}\fi
\ifnum\value{lstnumber}=5\color{purple!20}\fi
\ifnum\value{lstnumber}=8\color{purple!20}\fi
\ifnum\value{lstnumber}=10\color{purple!20}\fi}]
  ... // some calculation of %rax.1
  0xbf94: mov  $0x400000,%ecx.1
  |~~~~~~~~rcx.1 = 0x400000|
  0xbf99: cmp  %rcx.1,%rax.1
  |~~~~~~~~cf.1 = (rax.1 < rcx.1)|
  0xbf9c: jae  $0xcbb0
  0xbfa2: mov  %rdi.0,%rcx.2
  |~~~~~~~~rcx.2 = rdi.0|
  0xbfa5: add  %rax.1,%rcx.2 // dest = %rcx.3
  |~~~~~~~~rcx.3 = rcx.2 + rax.1|
  0xbfbb: mov  %esi.1, 28(%rcx.3)
  ... // argument handling, %rdi unmodified
  0xbff0: call bar
  ...
  0xcbb0: ud2
\end{lstlisting}

To understand why the memory access at \addr{bfbb} is safe, note that \addr{400000} is loaded into $\tt ecx.1$, which effectively clears the upper 32-bit of $\tt rcx.1$ and makes ${\tt ecx.1} = {\tt rcx.1} = \addr{400000}$. Given the conditional jump at \addr{bf9c}, the next instruction is executed only when ${\tt rax.1} < {\tt rcx.1}$, which is \addr{400000}. The code then stores the value of ${\tt rdi.0} + {\tt rax.1}$ to $\tt rcx.3$. Hence, the heap access at \addr{bfbb} is bounded by the heap region as its address $\tt rcx.3+28 = rdi.0 + rax.1 + 28$.
Recall that $\tt rdi.0$ is the heap base and the instruction is executed only when $\tt rax.1 < \addr{400000}$ ($\tt 4M$). 

The theorem prover verifies the memory access at \addr{bfbb} via the associated \vc[s], which is already validated per \autoref{sec:proof-processing}, as follows.
Recall that all axioms and validated \vc[s] are added to the set $\mathcal{F}$,
the constraint solver checks the following proof obligation at \addr{bfbb}: 
\begin{align*}
    \mathcal{F}\land \mathcal{PC} \rightarrow~ &{\tt rdi.0} < {\tt rcx.3} + 28 ~\land {\tt rcx.3} + 28 < {\tt rdi.0} + 8{\tt G}
\end{align*}
where $\mathcal{F} = ({\tt rcx.1}=\addr{400000})\land ({\tt cf.1}=({\tt rax.1}<{\tt rcx.1}))\land ({\tt rcx.2}={\tt rdi.0}) \land ({\tt rcx.3}={\tt rcx.2}+{\tt rax.1})$ and $\mathcal{PC}={\tt cf.1}$. The implication is obviously always correct.

Additionally, at \addr{bff0} ($\tt call~bar$), the value of the current $\tt rdi$ is checked against $\tt rdi.0$. Since $\tt rdi$ is never modified in this example, the heap base invariant across function calls is automatically verified.



\fi

\subsubsection*{Global Access Safety} 
The safety of global access is also similar to heap access, with a few changes. 
In Lucet-compiled code, the global memory region is a continuous space of 4KB. However, unlike heap memory, its base address is stored at a specific memory location: 32 bytes below the heap base.  To track which register stores the global base address, we introduce a distinguished symbol $\tt GB$.
\if false
------------ MODIFY ------------
Upon encountering an instruction of the form \verb|mov [Addr],Reg| (e.g., \verb|mov [%rdi-0x20],%rax|), which loads the global base address into a register, the BV validates the assertion ({\tt GlobalBaseLoc} {\tt Addr}) and adds it to the fact set $\mathcal{F}$.
To encode this event, Polich Checker provides a special derivation rule:
\[
  ({\tt GlobalBaseLoc}~{\tt Addr}) \rightarrow ({\tt GlobalBaseAddr}~{\tt Reg})
\]
------------ MODIFIED ------------
\fi
%
%
Upon encountering an \vc $\tt Reg=GB$ associated with an instruction $\tt mov ~[MemAddr],Reg$, the following proof obligation must be validated:
\[
  ({\tt rdi.0 - 32} = {\tt MemAddr}) \rightarrow ({\tt Reg}={\tt GB})
\]

To verify that a memory access instruction complies with global access safety, \name checks the following rule:
\begin{align*}
  \tt{GLOBAL\text{\ding{51}}} \equiv~ &{\tt GB} < {\tt Addr} \land {\tt Addr} < {\tt GB + 4{\tt K}}
\end{align*}
where ${\tt Addr}$ is the target of the validated memory access instruction. The rest of the verification process is the same as checking heap safety.
\autoref{code:global} shows a simple example of a global memory read and its safety \vc. At \addr{3fb0}, the heap base is loaded into $\tt r12.1$. The next instruction and its \vc can be validated through the implication above, letting the verifier know that $\tt rcx.1$ holds the global memory base $\tt GB$. Eventually, a global read is performed on the address $\tt rcx.2 + \addr{8}$, which is effectively ${\tt GB} + \addr{18}$. This is consistent with the $\tt{GLOBAL\text{\ding{51}}}$ rule, thus verifying the global access safety of this instruction.

\def\addrcolor{\color{ForestGreen}}
\def\asmcolor{\color{RoyalBlue}}
\def\assertcolor{\color{RubineRed}}
\begin{lstlisting}[label={code:global}, caption={Global access example and its corresponding assertions.}, float, escapechar=@, captionpos=b]
    @\aftergroup\addrcolor@       @\aftergroup\asmcolor@  -------- SSA-like Assembly -------- @\aftergroup\assertcolor@       ---- Assertions ----
    @\aftergroup\addrcolor@0x3fb0:@\aftergroup\asmcolor@  mov %rdi.0, %r12.1                  @\aftergroup\assertcolor@       r12.1 = rdi.0
    @\aftergroup\addrcolor@0x3fb3:@\aftergroup\asmcolor@  mov [%r12.1 - 0x20], %rcx.1         @\aftergroup\assertcolor@       rcx.1 = GB
    @\aftergroup\addrcolor@0x3fb8:@\aftergroup\asmcolor@  add 0x10, %rcx.1  // dest = %rcx.2  @\aftergroup\assertcolor@       rcx.2 = rcx.1 + 0x10
    @\aftergroup\addrcolor@0x3fbc:@\aftergroup\asmcolor@  mov [%rcx.2 + 0x8], %r13.1                                          
\end{lstlisting}

\subsubsection*{Optimization}
As discussed, each memory access can be verified via
%
$(\tt{HEAP\text{\ding{51}}} \lor \tt{STACKR\text{\ding{51}}} \lor \tt{GLOBAL\text{\ding{51}}})$ for memory read and $(\tt{HEAP\text{\ding{51}}} \lor \tt{STACKW\text{\ding{51}}} \lor \tt{GLOBAL\text{\ding{51}}})$ for memory write. 
As an optimization, the \vcgen can provide \emph{hints} on the accessed region and avoid the need of handling the disjunction above. 
Note that the hints are also untrusted, since given a wrong hint, the memory safety check fails.

\subsubsection{Control-flow Integrity}
\label{sec:sfi-cfi}
To ensure control-flow integrity (CFI),  the target of all control-flow transfers (i.e., jumps, calls, and returns) must point to valid code locations. Since direct calls and jumps are trivial to verify, we focus on indirect control flow transfers here. 

\subsubsection*{Indirect Call Safety}
While it is challenging to verify indirect call safety for an arbitrary binary, Lucet-compiled code follows a pattern that simplifies the verification. In particular, the target of each indirect call is retrieved from an entry in $\tt guest\_table\_0$
in its data section, and its starting address is unique for each binary. The table holds many 16-byte aligned (2-qword) entries, whose second qword holds a valid function pointer. The number of entries is stored at a distinguished address (part of the so-called $\tt lucet\_tables$), which is also unique for each binary.
Hence, verifying indirect calls boils down to checking if the call target is a valid entry in $\tt guest\_table\_0$.
%
%

\name adopts three special symbols from the policy specification and encodes their concrete values into the axioms, i.e., 
$\tt GT$ for the starting address of $\tt guest\_table\_0$, 
$\tt GTS$ for the number of entries in $\tt guest\_table\_0$,
and $\tt GTSAddr$ for the address holding $\tt GTS$.
Moreover, it inherits the predicate $\tt FnPtr$ such that the assertion $\tt FnPtr~Reg$ asserts that $\tt Reg$ holds a valid function pointer. 
The BV handles the special symbols and predicate as follows:

\begin{itemize}
    \setlength\itemsep{0em}
    \item It parses the binary being checked and link symbols $\tt GT$ and $\tt GTSAddr$ to their respective concrete values in the binary, for example, \addr{4000} and \addr{3008}: 
\begin{align*}
    \mathcal{F} = \{
                    {\tt GT}=\addr{4000}, {\tt GTSAddr}=\addr{3008}\}
\end{align*}

    \item To validate an \vc ${\tt GTS} = {\tt Reg}$ (i.e., $\tt Reg$ holds the number of entries in $\tt guest\_table\_0$) with its corresponding instruction $\tt mov ~ [MemAddr],Reg$, it checks ${\tt GTSAddr} = {\tt MemAddr}$,
    where ${\tt GTSAddr}$ is the concrete address storing the number of entries. 

    \item To validate an \vc $(\tt FnPtr~Reg)$ with its corresponding instruction $\tt mov ~ [MemAddr],Reg$, it checks
\begin{gather*}
  ({\tt GT} < {\tt MemAddr})  \: \land \: ({\tt MemAddr} < {\tt GT} + {\tt GTS} * 16) \: \land \\
   ({\tt MemAddr} \;{\tt BITAND}\; \addr{7} =  \addr{0}) \: \land \:
   (({\tt MemAddr} \;{\tt BITXOR}\; {\tt GT}) \;{\tt BITAND}\; \addr{8} = \addr{8})
\end{gather*}
where the first two clauses bound $\tt MemAddr$ within $\tt guest\_table\_0$, and the last two ensure ${\tt MemAddr}$ is the last 8 bytes of a 16-byte aligned entry.
\end{itemize}

Finally, the \regulation[initcap] Generator generates proof \regulation[s] for each indirect call instruction on $\tt Reg$ based on the following rule:
$    {\tt IND\_CALL\text{\ding{51}}} \equiv~ {\tt FnPtr~Reg}   $

\ifexample
\begin{lstlisting}[label={code:call}, escapechar=|,  caption={Indirect call example in SSA-like assembly with its assertions in colored background.}, linebackgroundcolor={
\ifnum\value{lstnumber}=3\color{purple!20}\fi
\ifnum\value{lstnumber}=5\color{purple!20}\fi
\ifnum\value{lstnumber}=7\color{purple!20}\fi
\ifnum\value{lstnumber}=10\color{purple!20}\fi
\ifnum\value{lstnumber}=12\color{purple!20}\fi}]
  ... // calculating %rcx.1
  0x6259: mov  $0x3000, %rsi.1
  |~~~~~~~~rsi.1 = 0x3000|
  0x6260: mov  8(%rsi.1), %rsi.2
  |~~~~~~~~rsi.2 = GTS|
  0x6264: cmp  %rsi.2, %rcx.1
  |~~~~~~~~cf.1 = (rsi.2 > rcx.1)|
  0x6267: jae  $0x6525
  0x626d: mov  $0x4000, %rsi.3
  |~~~~~~~~rsi.3 = 0x4000|
  0x6289: mov  8(%rsi.3, %rcx.1, 16), %rax.1
  |~~~~~~~~FnPtr rax.1|
  0x62a3: call %rax.1
  ...
  0x6525: ud2
\end{lstlisting}

\autoref{code:call} shows a simple example with indirect call, where the axioms specifies that ${\tt GT}=\addr{4000}$ and the value of $\tt{GTS}$ is stored at \addr{3008}.
To validate the \vc at \addr{6260}, the \regulation[initcap] Generator generates $(\addr{3008} = {\tt rsi.1+8})$ which is always true given the (validated) \vc at \addr{6259}.
To validate the \vc at \addr{6289} (i.e., the safety of indirect call), the \regulation[initcap] Generator first generates 
\[{\addr{4000}} < {\tt MemAddr} \land {\tt MemAddr} < {\addr{4000}} + {\tt GTS} * 16\]
where ${\tt MemAddr}={\tt rsi.3} + {\tt rcx.1}*16 + 8$. Given previously validated \vc[s] ${\tt rsi.3}=\addr{4000}$ and path condition ${\tt cf1.1}$, which is equivalent to ${\tt GTS}>{\tt rcx.1}$ based on the validated \vc[s] at \addr{6260} and \addr{6264}, it is easy to check that the condition above is always true. Moreover, ${\tt rcx.1}*16 + 8$ ensures that
the address is the last 8 bytes of a 16-byte aligned entry, hence, satisfying the last two checks of $(\tt FnPtr~Reg)$. Finally, at \addr{62a3}, the theorem prover verifies indirect call safety, which is directly implied by the checked \vc at \addr{6289}.
\fi

\subsubsection*{Indirect Jump Safety} 
Similar to indirect calls, Lucet compiles indirect jumps in a specific pattern so that verification is feasible. In this case, valid jump offset entries are stored in 4-byte aligned jump tables, baked directly into the code region, immediately after the indirect jump instructions. 

For each indirect jump, Lucet first calculates an index to its corresponding jump table (starting at $\tt JT$) and checks that the index smaller than the entry size. Next, the jump table is accessed at address $\tt JT + index * 4$ to load the offset, say $\tt off$ from $\tt JT$. Finally, an indirect jump instruction jumps to $\tt JT+\tt off$. To verify indirect jump safety, the policy specification introduces four predicates:
\begin{itemize}
    \item $\tt JT ~ Addr$, stating that $\tt Addr$ starts at \emph{a} jump table.
    \item $\tt JTS ~ Addr ~ s
    $, stating that for the jump table at $\tt Addr$, its number of entries is $\tt s$.
    \item $\tt JmpOff ~ Addr ~ off$, stating that $\tt off$ is a valid offset in the jump table at $\tt Addr$.
    \item $\tt JmpTgt ~ v$, stating that $v$ is a valid jump target.
\end{itemize}

Assume a binary with two jump tables at $\addr{64f3}$ and at $\addr{64c0}$ with 9 and 3 entries respectively, the BV first establishes the initial set of $\mathcal{F}$ as
\begin{align*}
    \mathcal{F} = \{
        &{\tt JT}~\addr{64f3},
        {\tt JTS}~\addr{64f3}~{\tt 9},  
        {\tt JT}~\addr{66c0},
        {\tt JTS}~\addr{66c0}~{\tt 3}
    \}
\end{align*}

To validate an \vc ${\tt JmpOff ~\addr{64f3}} ~{\tt Reg}$ (i.e., $\tt Reg$ holds a valid offset for the jump table at \addr{64f3}), a derivation rule handles the instruction $\tt mov~[MemAddr], Reg$:
\begin{gather*}
    ({\tt JT}~\addr{64f3}) \: \land \: (\addr{64f3} \leq {\tt MemAddr}) \: \land \: 
    ({\tt MemAddr} < \addr{64f3} + ({\tt JTS}~\addr{64f3}) * {\tt 4}) \: \land \: \\
    (({\tt MemAddr} - \addr{64f3}) ~ {\tt BITAND} ~ \addr{3} = \addr{0})
\end{gather*}
where the first clause verifies that \addr{64f3} is indeed a jump table head. The next two clauses ensure that the address is within the designated table. The last clause checks that the address is 4-byte aligned from $\tt JT$.

To validate the next \vc ${\tt JmpTgt ~ Dst}$ on an instruction $\tt add ~ Src, Dst$ equivalent to $\tt Dst' = Src + Dst$, another derivation rule is provided:
\begin{align*}
    &({\tt JT} ~ {\tt Src}) \land ({\tt JmpOff} ~ {\tt Src} ~ {\tt Dst})
\end{align*}

When encountering $\tt Jmp ~ Reg$ instructions, the Obligation Generator generates proof obligations according to the rule:
\[
    {\tt IND\_JUMP\text{\ding{51}}} \equiv {\tt JmpTgt ~ Reg}
\]

\subsubsection*{Return Address Safety}
The rule to be checked at a return site (assuming the current version of $\tt rsp$ is $\tt rsp.i$) is simply
\[
    {\tt RETURN\text{\ding{51}}} \equiv {\tt rbp.0 = rsp.i}
\]
The same methodology used in checking stack safety can be applied, as both are mainly handling $\tt rsp$ and $\tt rbp$.

\subsection{Verifying Information Flow Safety}
\label{sec:example-ifc}

Next, we demonstrate how to offload the majority of an information flow analysis to untrusted parties, while ensuring information flow security using \name. One complexity of an information flow analysis comes from reasoning about memory, which tracks data flow via memory alias. For example, a static information flow analysis called PIDGIN~\cite{johnson2015exploring} consists of approximately 22.7 KLOC (7.5 KLOC of which for alias analysis) on top of the WALA framework~\cite{wala}.

\begin{figure}[tb]
    \centering
    \includegraphics[width=0.6\columnwidth]{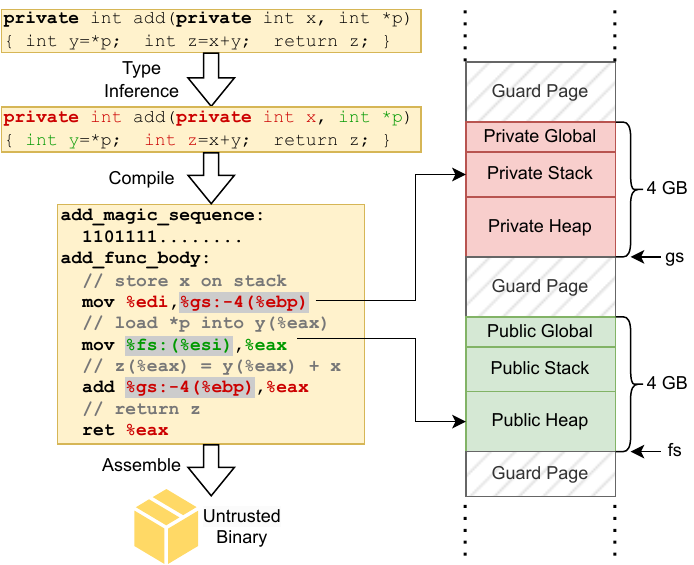}
    \Description{Description placeholder}
    \caption{Illustration of \confllvm's workflow.}
    \label{fig:confllvm}
\end{figure}

\name leverages an LLVM-based information flow analyzer, \confllvm~\cite{confllvm}, as the assertion generator and extracts its formally verified specification.
\confllvm's verifier targets untrusted x86-64 binaries compiled from its toolchain, as illustrated in \autoref{fig:confllvm}.
The input is a C program where each function comes with security annotations (i.e., $\tt private$ type qualifier) on arguments and return values.
During compilation, 
\begin{itemize}
\item Each function's annotation is encoded into a 64-bit \textit{magic sequence} preceding the function. The first 7 bits of a magic sequence encodes the secrecy of return value and 6 register-passed arguments. The rest bits are irrelevant for IFC. For the $\tt add$ function in \autoref{fig:confllvm}, the first 7 bits are 1101111 as the return value and the first parameter are private thus set as 1. The other 4 bits corresponding to unused arguments are conservatively set as 1. 

\item \confllvm uses static taint analysis to propagate taints within each function, and then reallocates all memory objects (in stack, heap and global space) into two disjoint memory regions according to objects' sensitivity. 
The memory regions are masked by two 4GB-aligned segment registers $\tt fs$ and $\tt gs$, so all tainted memory objects are stored and loaded from the private ($\tt gs$ masked) memory segment, through the address format $\tt gs:[exx+disp]$, where $\tt exx$ is 32-bit general purpose register (GPR) and $\tt disp$ is 32-bit displacements. 
Public data is accessed via $\tt fs$. 
\end{itemize}

\subsubsection{Information Flow Semantics and Constraints}
Unlike SFI analysis, an information flow analysis reasons about the secrecy of registers/memory locations, rather than their values. Hence, we interpret all variables in the assertions ($\bnf{var}$ in \autoref{fig:proof}) as their sensitivity, where $\true$ and $\false$ represent secret and public respectively. Moreover, we abstract the IR semantics (\autoref{sec:ssa}) so that information flow is encoded as a logical implication ${\tt src} \Rightarrow {\tt dst}$, replacing the value-based IR semantics where ${\tt src}$ and ${\tt dst}$ are the source and destination of an instruction. All memory loads/stores to segment registers $\tt fs$ and $\tt gs$ are modeled as public and secret respectively; other memory instructions are disallowed.

IFC requires that secret information never ``flows'' to public locations (i.e., $\true\Rightarrow \false$ never occurs). As standard, a proof obligation in IFC follows the form of $\exists \vec{v}. ~ \mathcal{F} \land \mathcal{S}$, where $\vec{v}$ denotes variables whose secrecy is to be solved. Hence, IFC is a must policy, requiring a \sat model.

\subsubsection{Information Flow Control}
In \confllvm compiled binaries, both magic sequences and segment register usages are assertions in \name's framework: they claim the secrecy of a function's parameters and return value, as well as memory loads and stores. The assertions require further validation as they can be crafted to leak sensitive information, e.g., by switch a memory write from the public segment to the private segment, or to switch an argument's taint bit in the magic sequence from 0 to 1. 




To validate a magic sequence of a function $f$, we first check information flow restrictions at each of its call site: we ensure that the parameters' secrecy can flow to the corresponding taint bit of $f$. Right after the call instruction, all caller-saved registers (expect the return registers) are conservatively set to be private, whereas the return registers' status (e.g., $\tt rax$, $\tt xmm0$, etc.) are set according to $f$'s magic sequence. For instance, if a caller calls \texttt{add} in \autoref{fig:confllvm} with actual parameters $\tt edi.2$ and $\tt rsi.1$ and returns ${\tt eax.3}$, then validation requires $\mathcal{F} \land {\tt edi.2} = {\tt true} \land {\tt rsi.1} = {\tt false} \land {\tt eax.3} = {\tt true}$, where $\mathcal{F}$ is the the fact set at the call site (\autoref{sec:proof-processing}).
Moreover, at the beginning of $f$, we use the taint bits in the magic sequence to taint the secrecy labels of all corresponding parameter-passing registers, and then validate the taint bit on its return value at the return instruction (e.g., in \autoref{fig:confllvm}, we validate $\mathcal{F}\land {\tt eax.2} \Rightarrow {\tt true}$ at return).

Validating segment register usage is easier: we simply validate that the segment's secrecy can flow to the target register (for a load instruction), or the source register can flow to segment's secrecy (for a store instruction).

\subsubsection*{Policy specification and obligation generation}
\confllvm specifies IFC policy as a list of \emph{trusted} functions. These functions are used to specify secret sources (e.g., the initial allocations of a key), public sinks (e.g., returning to the public domain), as well as declassification functions that might reveal secret-dependent information (e.g., an encryption function). These functions are not checked by \name; they serve as the root of trust in IFC. 
To verify that no IFC violation occurs, the obligation generator ensures that for every public parameter of a trusted function, it only takes public values on its call site. Moreover, to control implicit flows, 
the obligation generator also checks that all branch conditions must be public. 

\subsection{Verifying LVI Mitigation}
\label{sec:example-lfence}
Last, we target Load Value Inject (LVI) attack which exploits a micro-architectural load buffer to inject a value after a faulting load~\cite{lvi}.
Inspired by the methodologies of modeling micro-architecture semantics~\cite{mcilroy2019spectre, cauligi2020constant},
we extend the \vc language to represent the behaviors of the load buffer.
A new $\tt flag$, $\tt LoadBuffer$ is introduced to represent if the load buffer is filled. We also extend the semantics of setting $\tt LoadBuffer$ for memory load instructions.
It is cleared by $\tt lfence$,
as all loads before this instruction retire such that the load buffer is no longer occupied. 

Practical mitigations have been integrated into GNU \texttt{binutils} (i.e., the \texttt{as} assembler) and further adopted by Intel's SGX SDK.
The policy enforces the instrumentation strategy purposed by \texttt{as}, which inserts $\tt lfence$ instructions immediately after load instructions.
Therefore, for every sensitive load instruction (e.g., those after branches), the Obligation Generator introduces proof obligations that $\tt LoadBuffer$ at load's next instruction is cleared. 

\if lvi

\subsection{Verifying LVI Mitigation}
\label{sec:example-lfence}
We now elaborate on the Lfence policy
to demonstrate the extensibility of \name, particularly with regard to side-channel mitigations.
We target Load Value Inject (LVI) attack which exploits a micro-architectural load buffer to inject a value after a faulting load~\cite{lvi}.
Inspired by the methodologies of modeling micro-architecture semantics~\cite{mcilroy2019spectre, cauligi2020constant},
we extend the proof language to represent the behaviors of the load buffer.
A new $\tt flag$, $\tt LoadBuffer$ is introduced to represent if the load buffer is filled or not, and it is set whenever there is a memory load.
It is cleared by $\tt lfence$,
as all load instructions prior to $\tt lfence$ instruction retire such that the load buffer is no longer occupied. 

Practical mitigations have been integrated into GNU \texttt{binutils} (e.g., \texttt{as} assembler)~\cite{binutils} and adopted by Intel SGX SDK~\cite{sgx:sdk}.
The Lfence policy verifies that binary applies \texttt{-mlfence-after-load}, an instrumentation strategy purposed by \texttt{as}, which inserts $\tt lfence$ instructions immediately after load instructions
Therefore, for every load instruction, Polich Checker asserts that the $\tt LoadBuffer$ at its next instruction is cleared. 

\fi
\section{Bounty Task Manager}
\label{sec:btm}

\ignore{
\hongbo{New Outline
Why we do not trust SMT solvers? 1. Alternative solution; 2. Benefits of delegating constraint solving 3. Lower bound
Now we introduce the difference between verification service and Bug bounty Program.
1. Compare the form of bugs (\sat model)
2. Problems can be solved by the smart contract: automatic validation and TTP => solved by the smart contract-based solution
3. Issues originated from verification service: requiring no-bug indication, security assurance. Plus the probability of finding bugs is low.
See \autoref{tab:comparsion} for reference.
}

\begin{enumerate}
    \item Why do we not trust SMT solvers? What are the shortcomings of it? \textit{- This is pretty much already said, first paragraph of 6.0.}
        \begin{itemize}
            \item Present the alternative solution, a \emph{TRADITIONAL?} bug hunting paradigm. 
            \item Benefits of delegating constraint solving - open and auditable
            \item Lower bound \textit{- This is trustworthiness, but I think this should go to a later point where BTM has been introduced}.
        \end{itemize}
    \item What is the problem with the simple bug hunting solution? Introducing the BTM the solve the problems.
    \item Compare the difference between a traditional bug hunting program and the BTM solution. See \autoref{tab:comparsion} for reference.
        \begin{itemize}
            \item The form of bugs (\sat model)
            \item Problems can be solved by the smart contract: automatic validation and TTP $\Rightarrow$ solved by the smart contract-based solution
            \item Present the challenges originated from verification service
                \begin{itemize}
                    \item requiring no-bug indication
                    \item security assurance
                    \item Plus the probability of finding bugs is low
                \end{itemize}
        \end{itemize}
\end{enumerate}
}

The Bounty Task Manager (BTM) in \name solves the constraints from the BV, fostering an inclusive and transparent verification environment while maintaining a minimal TCB. 
It deals with both \sat and \unsat constraint files for $\exists$- and $\forall$-quantified policies, respectively.
Given that validating \unsat results is more challenging, we elaborate on it first.

The BTM borrows the idea from bug bounty programs and delegates heavyweight theorem proving tasks to untrusted bug bounty hunters (BBHs), validating their solutions via a lightweight protocol.
This approach significantly reduces the TCB size.
Furthermore, BBHs can utilize a range of frameworks, including translations between SMT constraints and other logical representations~\cite{besson2011flexible,armand2011modular}, prompting openness.

Such an open ecosystem welcomes all solving methods, including theorem provers from industry and academia, experimental algorithms, and even manual inspection,  to contribute.



However, adopting traditional bug bounty programs directly does not work for a couple of reasons.
Their dependence on trusted third parties for bug confirmation and reward distribution conflicts with \name's goal of openness.
To mitigate this, we adapt smart contract-based bug bounty designs which leverage TEE for cheaper computations~\cite{tramer2017sealed,cheng2019ekiden}, automating the validation of bug submissions and the issuance of rewards.
Therefore, \name completely removes trusted third parties.
Moreover, by integrating with the blockchain, \name enhances its transparency to a higher level.
The \vc[s], certified binaries, SMT files, and verification results could all be published on the blockchain, allowing full replication of the verification process.
Users can also independently audit the CC applications and the smart contracts via the primitives of TEEs and blockchain.

\ignore{
\subsubsection*{Challenges}
Traditional bug bounty programs still rely on trusted third parties, since the confirmation of bugs and rewarding involves human intervention.
To remove trusted third parties, \name alters prior work built atop the smart contract~\cite{tramer2017sealed,cheng2019ekiden}.
Bugs can also be confirmed mechanically by the smart contract\hongbo{do we need to state that smart contract needs automatic validation of bugs? This may not always be true?}.
This is straightforward, as re-evaluating the constraints with the valuation in the model achieves bug validation.
Moreover, smart contracts can automatically pay rewards to BBHs after confirming submitted bugs, addressing the concerns of \BBH{}s regarding bug-reporting and compensation~\cite{akgul2023bug}.
These features base \name's constraint solving on a bug bounty program using a smart contract.\danfeng{I don' think from "This is straightforward ...." adds too much too the paper. Maybe we can just say "the adoption is mostly straightforward and we discuss more details in the implementation (or appendix?)}
}

\subsubsection*{Challenges}
More significant challenges originate from the gap between verification service and bug bounty programs:
\name must ensure verification results to users beyond accepting bug submissions and validating bugs.
Bug bounty programs fail the requirement in two important ways.
First, bug bounty programs focus solely on identifying bugs, neglecting ``no-bug'' submissions pivotal for verification.
More specifically, when BBHs do not submit anything, the BTM must discern further between two scenarios: BBHs have no interest in the task, vs. BBHs try hard but fail to find a bug (i.e., the binary is verified by BBHs).
Bug bounty programs ignore such concerns since they \emph{do not} aim to warrant any security guarantee on software. 
%
%
Therefore, the BTM needs to measure BBHs' efforts in bug finding (i.e., constraint solving) even if no bug is found (\textit{Challenge 1}).
Second, \xchanges{the BTM should provide (quantitative) assurance to the verification results (\textit{Challenge 2}), requiring the protocol to not only incentivize ample BBHs but also manage worst-case scenarios.} 


\ignore{
\done\danfeng{we have not introduced payments yet}, addressing the concerns of \BBH{}s regarding bug-reporting and compensation~\cite{akgul2023bug}. \sen{can smart contract publish verification results? we also mention this challenge in a later paragraph}
\fanz{this paragraph gives readers the impression that our innovation is to introduce smart contract-based bug bounties, which is not the case. I think in this paragraph, we should  establish why bug bounty in general is useful for us (without delving into how to run a bug bounty, with human or smart contracts). Hopefully this paragraph is not too long. the next paragraph should be merged with this one, both establishing the benefits of bug bounty.}
\quan{I think one thing we failed to establish is that the direct target of the bug hunting program are \emph{bugs in SMT solvers, not binaries to-be-verified}. That in my opinion is subtle, but should be pointed out to aid explaining the motivation of the complex task bundling and rewarding system (i.e., C2 and C3).} \hongbo{mention it in the challenge?}

\ignore{
In the last step of binary verification~\autoref{sec:satcheck}, the BV forwards \unsat constraint files to the BTM to validate them via a crowdsourcing bug bounty protocol.
This protocol allows anyone, referred to as a {\em bug bounty hunter} (\BBH{}), to find and report bugs in exchange for monetary rewards, using \textit{any} means, such as running existing SMT solvers, developing new algorithms, and even manually solving.
Here, a bug refers to a satisfiable valuation of the variables (i.e., a satisfiable model) in constraints but deemed unsatisfiable by untrusted SMT solvers.
Again, validating bugs submitted by BBHs requires a much smaller TCB than typical SMT solvers.
}

\ignore{
Validating an unsatisfiable result is hard, whereas validating a satisfiable model is easy\fanz{just ``easy''}.
Therefore, we substitute the complex and huge SMT constraint solvers with a simple model validator and related components\fanz{what is a driver in this context?} in the TCB.
In the BTM, ``bug'' refers to a satisfiable model\fanz{this definition should appear way earlier.}, and finding a bug means there is a valuation (i.e., model) satisfying all logical predicates, which implies policy incompliance.
}

\subsubsection*{Challenges}
To remove trusted third parties, previous work has employed smart contracts to build bug bounty programs, curtailing expenses of smart contract-based computation by harnessing off-chain TEEs~\cite{tramer2017sealed,zhangTown2016,cheng2019ekiden}.
However, \name requires additional capabilities beyond simply accepting bug submissions, facing several unique challenges. \done\fanz{It's not that their methods do not fit \name, it's that \name requires additional properties that a bug bounty (traditional or smart contract based) cannot provide.},
On the one hand, as a verification service, \name must deliver a stronger assurance,
whereas traditional bug bounty programs focus only on discovering potential vulnerabilities by third-party participants. In other words,
\name needs to publish trustworthy \hongbo{quantitative} verification results while the BBHs are untrusted.
More specifically, when there are no bug submissions, \name needs to differentiate further between two scenarios: BBHs have no interest in the task, vs. BBHs try hard but fail to find a bug (i.e., the application is verified), whereas traditional bug bounty programs do not care since they \emph{do not} aim to provide security guarantee to the software provider\quan{provide guarantee to software providers?}. 
%
%
Therefore, similar to {\em proof of work} mechanism~\cite{nakamoto2008bitcoin}, the BTM should appreciate BBHs' efforts in constraint solving even if no bug is found (\textbf{C1}) and use the efforts to gauge the assurance on the software being verified (\textbf{C2}).
On the other hand, another unique challenge is that the probability of finding bugs in the bounty tasks is slim, as the reference solvers have already filtered many policy-incompliant binaries (see \autoref{sec:satcheck}).
Undetected bugs indicate flaws in reference solvers.
Therefore, the bounty structure demands a careful setup to incentivize a sufficient number of BBHs (\textbf{C3}).

\hongbo{ignore the rest of this section to \autoref{sec:btm-solution}, maintaining for now to recheck the comments}

One intuitive approach involves deploying a smart contract that archives the constraint files and welcomes submissions from any \BBH{} participant.
This contract compensates the submitter with a preset reward, marking the relevant binary as insecure if the submission contains a validated bug (i.e., the submitted model satisfies the constraints).\danfeng{we need to try a bit harder before this to convince the reader that smart contract is a good fit for our goals. Maybe one approach is to first say bug bounty programs are appealing to achieve our goals. But a direct application will not work.}
Conversely, a binary receives a ``secure" designation if no vulnerabilities are detected within a duration.

Despite the potential to curtail validation expenses by harnessing TEEs for off-chain computations~\cite{zhangTown2016,cheng2019ekiden},
this straightforward design still confronts two primary obstacles: the economic motivation and the invisibility of BBH endeavors.
Firstly, the arduous nature of identifying satisfiable models may deter BBHs from participating actively. 
Secondly, the efforts put forth by BBHs go unnoticed, as \unsat verification results come with no evidence of working if no bug is found. \done\danfeng{be more specific?}.
As a verification service, \name ought to relay to users the verification results via the smart contract, detailing the number of attempts a BBH made to verify a binary.
However, users receive no additional feedback from the smart contract when no bug exists, both when BBHs are struggling to find satisfiable models or they simply lack interest, confusing them on whether to trust the binary.
\done\danfeng{this does not sound too bad. Why it matters? Maybe we should say the challenge is to create a proper incentive mechanism, while allowing untrusted participants? Anyway, it was not very clear what are the main technical challenge there at this point.}.
}

\ignore{
As for guarantee, the verification results are sound enough, \name faces two types of attackers.
Some malicious BBHs may try to let a bug remain undetected and make the verification pass, e.g., by submitting falsified constraint-solving results.
System administrators may also render the TEE (or server) offline and cease the submission of potential new bugs.
}

\ignore{
A natural idea is to deploy a smart contract that stores the constraints to be verified and accepts submissions of satisfiable models from anyone working as \BBH{}.
The smart contract pays the submitter a predefined reward and flags the corresponding binary as insecure if the submission contains a validated bug (i.e., the submitted model satisfies the constraints).
A binary will be flagged as secure if no bugs are found within a predefined time frame.

While conceptually correct, the above idea is impractical as storing large constraints and verifying satisfiability using a smart contract is prohibitively expensive. For example, performing 10 million additions in an Ethereum smart contract costs about \$1,110 as of July 2023, while a desktop CPU can compute billions of additions per second.
To reduce the computation cost, one can utilize TEEs to perform computation {\em off-chain}~\cite{cheng2019ekiden} and have smart contracts verify the computation results.
Specifically, to claim the reward, a \BBH{} with a bug $B$ runs a predefined validation program in a hardware TEE and submits the validation result with an attestation report~\cite{intel:sgxexplained}.
This can prove that $B$ indeed satisfies the constraints.

The smart contract needs to verify the attestation instead of evaluating SMT constraints.
Since constraints need not be accessed by the smart contract, it can be stored in a distributed file system (e.g., IPFS \cite{ipfs}) to reduce storage costs.
\ignore{
Then the \BBH{}s can query the tasks from the smart contracts, fetch the constraints, and try to submit their results to the \sat Validator.
The \sat Validator will validate the submissions and update the smart contracts if the \BBH{}s find a bug.
}

However, despite the improvements, this solution still faces three concerns.
First, verifying the attestation by smart contract is expensive as of July 2023 (about \$124) and we should reduce the cost.
\hongbo{move this to later analysis: since constraints are considered unsatisfiable by untrusted SMT solvers (otherwise, they would have been rejected), the probability of them being satisfiable is small, which suggests a bug in SMT solvers}
Second, the bug bounty program should promise profits for \BBH{}s to attract them as monetary compensation is a definitive factor~\cite{akgul2023bug}.
Third, proof of work is necessary for the bug bounty protocol.
If no bugs are submitted, there is no way to tell if there isn't any bug or if no \BBH{}s were attracted to work on the tasks. 
\ignore{
We note these security concerns prevent us from achieving the design goals of the BTM
To effectively overcome these concerns, our solution needs to: 1) ensure the \BBH{}s can submit even if the BTM is offline; 2) incentivize the BBHs facing hard constraint-solving problems; 3) \fanz{unclear what does this mean} \sen{remove diversify} ensure the submissions can enhance the confidence of policy compliance for the binaries.
}
}

\begin{figure}[!t]
	\centering
        \includegraphics[width=0.8\textwidth]{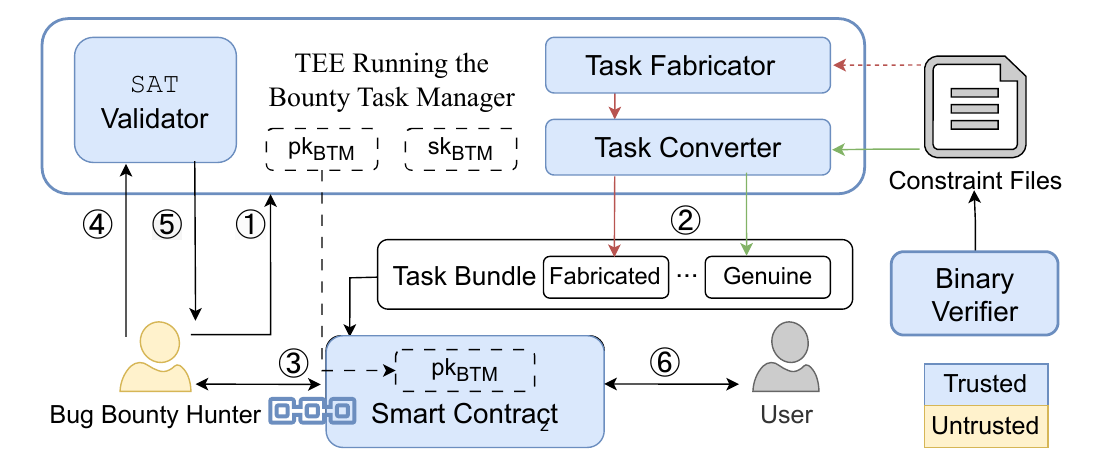}
        \Description{Description placeholder}
	\caption{BTM Workflow for $\forall$-quantified policies.}
	\label{fig:btm-workflow}
\end{figure}

\subsubsection*{Bug Bounty Protocol}
Before diving into the technical details of addressing the challenges, we first introduce the high-level workflow of the BTM 
according to \autoref{fig:btm-workflow}.
Initially, the BTM receives constraint files from the BV (\autoref{sec:proof-checking}).
Upon a BBH's request (\ding{172}), the \textit{Task Fabricator} and the \textit{Task Converter} collaboratively generate a {\em task bundle} containing multiple constraint files (\ding{173}) and publish it via the smart contract.
The BBH can then fetch the task bundle from the smart contract (\ding{174}).
After solving a task bundle, the BBH submits their answers for every task (i.e., either \unsat or a \sat model) to the BTM (\ding{175}), which then validates the answers using the \sat \textit{Validator}, allocates a reward (\ding{176}), and updates the verification results on the blockchain accordingly.
Finally, users can query the verification results of a binary (\ding{177}),
determining whether to use the binary according to the results.
We explain next how the aforementioned challenges are resolved using the building blocks in \autoref{fig:btm-workflow}.
A formal protocol with detailed algorithms is presented in~\autoref{appx:btm-full-workflow}.

\subsection{Challenge 1: Indication of Efforts}
\label{sec:indication-of-efforts}
Both bug and ``no-bug'' submissions should be validated as submissions are not trusted.
For $\forall$-quantified policies expecting \unsat, ``bug'' is a \sat model, whose successful validation via \sat Validator indicates policy violation.
However, validating \unsat results poses a significant challenge\footnote{Recent advancements in verifiable proof certificates alleviates the problem to some extent~\cite{armand2011modular, ekici2017smtcoq, andreotti2023carcara, barbosa2022flexible}. However, not all solvers/algorithms can emit proof certificates,
and existing techniques lack support for some theories: \texttt{SMTCoq} only support quantifier-free theories~\cite{ekici2017smtcoq}, and veriT can fail to generate checkable proof involving quantifiers~\cite{andreotti2023carcara}.}.


To tackle the challenge, we purposefully generate ``fabricated tasks'' (for which we know the answers) and blend them with genuine tasks (which are \ndsschanges{associated with verification tasks} from the BV) as a \emph{task bundle}.
BBHs fail on fabricated tasks with a high probability if they do not spend effort working on \emph{all} tasks in this bundle.
Therefore, by requiring submissions of answers to the entire task bundle, the BTM can confirm, through a layer of indirection, that a BBH spends effort on all tasks.
Considering the probability of finding bugs is slim thus genuine tasks are likely to produce \unsat when solved, a task bundle should have the following properties:

\begin{enumerate}[(i)]
    \item Fabricated tasks serve as an indication of work, forcing BBHs to solve all tasks in a bundle, as they must provide the known answers to all fabricated tasks.
    \item Fabricated tasks should be indistinguishable from genuine tasks \ndsschanges{in polynomial time}. Otherwise, BBHs can simply submit \unsat to all identified genuine tasks without solving them with any effort.
    Additionally, the BBHs should not receive repeated tasks, since they can reuse previous answers on them. 
    \item BBHs should not know the exact number of fabricated tasks with \sat results in a bundle; otherwise, they can solve tasks till the count of \sat answers is met and then answer the rest as \unsat without solving them.
\end{enumerate}

\ignore{
To tackle the challenge, our insight is to \emph{inject fabricated tasks} and mix them with the {\em genuine tasks} as a {\em task bundle}, which is delivered to BBHs.
Genuine tasks are from the constraint files generated by the BV; they correspond to functions subject to verification and presumably to be \unsat.
Fabricated tasks are crafted via Task Fabricator and designed to be \sat.
\danfeng{why are the designed to be \sat? what goes wrong if the fabricated ones can both \sat and \unsat?}.
\hongbo{we should guarantee the satisfiability of fabricated tasks is known correctly to the BTM, and we cannot guarantee it for \unsat.}
\fanz{this is a good question. where is the answer?} \hongbo{this can be explained later after we explain the reward structure, but I don't know how to explain it thoroughly here. We could say that, because \sat results come with \sat models, which can be validated individually within the TCB but \unsat results come with no evidence, we can only guarantee the production of \sat tasks.}\danfeng{For \unsat fabricated tasks, if BBH says unsat, it is legitimate. if BBT says sat with a model, we can check it. I don't see what's the real issue with this.} \hongbo{validating these \unsat fabricated tasks are the same as genuine tasks. However, they are not necessary, as verifying them does not contribute to the verification of functions but introduces additional overhead. Imagine if all fabricated tasks are \unsat, BBH can just answer \unsat for all tasks and bypass solving. Therefore \sat fabricated tasks are necessary, and the introduction of \sat fabricated tasks renders \unsat fabricated tasks useless.}
}


\begin{figure}[!t]
    \centering
    \includegraphics[width=0.7\textwidth]{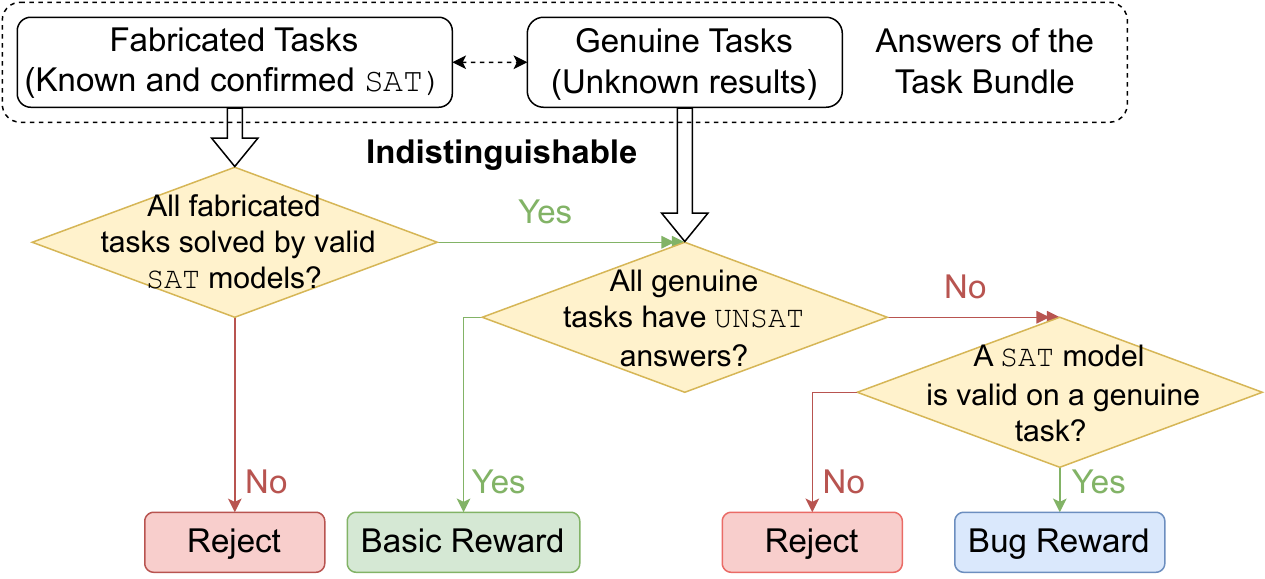}
    \Description{Description placeholder}
    \caption{Validation of the answers of a Task Bundle.}    
    \label{fig:answer-tree}
\end{figure}

To satisfy these properties, we introduce several components and techniques here. 
For property (i), 
the Task Fabricator is introduced to carefully craft constraint files with \sat models as answers, \ndsschanges{which could be validated by the \sat Validator}.
One way to generate fabricated tasks is \ndsschanges{randomly} sampling from the constraint set of a constraint file received from the BV.
The \sat Validator \ndsschanges{assures} the answers of these tasks are validated \sat models.
\autoref{fig:answer-tree} shows how the BTM validates the submission from a BBH.
First, the submission is only considered legitimate if all fabricated tasks are solved with valid \sat models, indicating that the BBH must have solved all tasks in the bundle (with the aforementioned properties).
Next, if all genuine tasks are answered \unsat, the smart contract allocates the BBH a {\em basic reward} and increases the verification counts of functions corresponding to all genuine tasks by one.
Otherwise, if a \sat model is reported for any genuine task, the BTM validates the model and, based on the result, either marks the corresponding binary insecure and issues the BBH a {\em bug reward}, or rejects the whole answer set.

To satisfy property (ii),
the Task Converter converts each task to a series of equivalent formulae~\cite{huth2004logic} (i.e., the original and converted constraints have the same satisfiability) to prevent reuses of answers.
Since no polynomial-time algorithm can confirm the equivalence of two logical proposition sets (i.e., a series of constraints)~\cite{eiter2004simplifying, eiter2007semantical}, the Task Converter essentially ``obfuscates'' \ndsschanges{constraint files before and after the conversion.
Thus, BBHs cannot avoid constraint solving.
Moreover, the conversion also breaks the link (i.e., subset inclusion) between the constraints received from the BV and fabricated tasks derived from them, making them indistinguishable in polynomial time.
Nevertheless, breaking such indistinguishability assumption cannot prevent benign BBHs from receiving task bundles containing buggy tasks.
Thus, new bugs can still be uncovered and submitted.}
Property (iii) is met by randomizing the number of fabricated tasks in each task bundle. Internally, the BTM records which tasks are fabricated and genuine; but the information is never revealed to the BBHs.

\ignore{
\subsubsection*{Validation of Submission}
\autoref{fig:answer-tree} demonstrates the answer validation and reward allocation according to the submission to the task bundle.
Firstly, as mentioned earlier, the \sat Validator validates \sat model for all fabricated tasks, rejecting the submission otherwise.
As \sat model implies a bug in genuine tasks, a validated \sat indicates the task corresponds to a buggy function, the BTM marks the corresponding binary insecure and grants the BBH a premium {\em bug reward}.
If all genuine tasks receive \unsat answers, the smart contract allocates the BBH a {\em basic reward} and increases the verification counts of functions corresponding to all genuine tasks\done\fanz{have we defined verification count?} by one, implying \fanz{not necessarily implying this. it is possible to claim basic reward without checking genuine tasks} the genuine tasks in this bundle are also examined once but no bug is found (i.e., validated as \unsat).
\hongbo{submit bug after reward?}
}


\ignore{
\subsubsection*{\textit{Challenge \textbf{C2}--Indication of Efforts}}
In the open verification ecosystem on the blockchain, it's essential to provide a {\em proof of solving} (PoS) to both display verification results and disperse rewards.
Hence, the BTM must validate the results of constraint solving as BBHs and reference solvers are not trusted. \done\fanz{we have not introduced BTM and the overall architecture yet }
Although there is no trivial way to verify \unsat results, it is worth noting that validating \sat results becomes feasible by re-assessing the satisfiable model.
Providing satisfiable models can then act as the PoS for the bug bounty program, enabling indirect validation to \unsat results via the tasks bundle. \done\fanz{i think you want to say we are going to design a protocol to do this. 
}
As explained in \autoref{fig:answer-tree}, the BBHs must solve all fabricated tasks to qualify for reward.
As Task Converter obfuscates the constraints, BBHs cannot take advantage of other information (e.g., previously solved task bundles) to avoid solving.
Therefore, the optimal strategy for rational BBHs seeking for reward is conducting constraint solving faithfully.
}

\subsection{Challenge 2: Verification Result Assurance}
\label{sec:trustworthy-verification}

\ignore{
Both maintaining the smart contract and the bug bounty program require monetary costs.
The computation power of smart contracts is extremely limited and expensive, so it is formidable to validate constraint solving on the chain.
Thus, we apply an off-chain solution~\cite{} to validate the satisfiable model in the enclave operated by the BTM, as the BTM can attest itself to the public.
}
\ignore{
\begin{table*}[!htbp]
    \caption{The attack vectors of the BTM and our mitigation. \hongbo{do we remove it?}}
    \centering
    \footnotesize
    \resizebox{\textwidth}{!}{
    \begin{tabular}{p{2.5cm}|p{5.5cm}|p{5.8cm}|p{3cm}}
    \toprule
       \textbf{Entity} & \textbf{Attack} & \textbf{Affect} & \textbf{Mitigation} \\
       \hline
       Administrator & The administrator makes the official BTM TEE offline or blocks the submission. & Breaking soundness. If a bug is found but cannot be submitted, it will harm the soundness. & Censorship-resistant submission channel \\
       \hline
       \multirow{2}{*}{BBH} & Lazy BBHs only solve fabricate tasks and claim basic rewards. & Breaking economic sustainability. They claim rewards without contributing any verification results.  & Incentive design. Proved in~\autoref{lemma:solving}. \\
       \cline{2-4}
       & The adversary controls a group of BBHs, hoards many task bundles, and deliberately reports buggy tasks to be \unsat (i.e., free of bugs).
  & Breaking soundness. They may make a buggy binary pass the verification process. & Incentive design and honest BBHs participation. Proved in~\autoref{thm:buggy-program-launderer}. \\
       \bottomrule
    \end{tabular}
    }
    \label{tab:attack-vectors}
\end{table*}
}


\ignore{
\subsubsection*{Deployment}\danfeng{the flow breaks here. Maybe first introduce the major components, which are more interesting, and finally touch the deployment details.}
The BTM is hosted in TEE and is managed by an untrusted administrator.
Upon initialization, the BTM produces a key pair and publicizes its public key.
Following this, the \textit{smart contract} (\smartcontract) is activated, integrated with the public key and the BTM's measurement, ensuring exclusive write privilege for the BTM to update the \smartcontract's states.
Thus, the verification results recorded on the blockchain are tamper-proof.
The BTM and the BV establish a trusted channel after remote attestation.
This connection guarantees that the BTM receives the constraint files from a trustworthy source, and the BV can publish signed binaries on the blockchain.\hongbo{make the description consistent}
In this way, BBHs and binary users can conduct remote attestation with \name and check the key and measurement are consistent with those recorded on the blockchain, assuring the BTM is not compromised.
}


As a verification service, \name must ensure trustworthy results to users, which is based on two premises: 1) the integrity and availability of \name itself, and 2) participation of enough BBHs.

\subsubsection*{Integrity and Availability}
As the BTM and the BV are deployed in TEE accompanying the on-chain smart contract, everyone can conduct remote attestation and inspect the on-chain results, ensuring the whole workflow is not compromised. 
\name also considers DoS attack.
If bug submission is suppressed, users will not be aware of newly submitted bugs on the so-called ``verified binary''.
Note that BBHs submit bugs to the BTM, which marks the corresponding binary as insecure on the smart contracts after confirmation.
Therefore, the malicious system administrator of the BTM can suppress (i.e., censor) bug submissions, threatening the correctness of verification results.
Thus, an alternative censorship-resistant submission channel is essential.
To do so, BBHs can spin up a local BTM instance in a BBH-controlled TEE node.
Note that the local BTM is equally trustworthy from the point of view of \smartcontract{} thanks to attestation (Appendix~\autoref{sec:censorship-resistant-bug-submission}).
Thus, the local BTM can perform validations of tentative bugs, informing users via the \smartcontract{} that censorship is happening and the binaries under verification may contain unsubmitted bugs.

\subsubsection*{Participation of BBHs}
\name also needs to solicit participation from enough BBHs.
To guarantee a lower bound of the assurance, \name deploys some untrusted ``reference BBHs'', each running a unique SMT solver.
Therefore, even in the worst-case scenario where no other BBH participates, the resulting system is at least as trustworthy as distributing trust across those reference SMT solvers, which is \ndsschanges{at least as secure as} traditional verification workflows.

Nevertheless, the participation of other BBHs will strictly elevate the trustworthiness of the system.
To attract an adequate volume of BBHs to submit their solutions, our protocol theoretically guarantees positive \ndsschanges{expected} returns for the BBHs.
We set the basic reward to the computational cost of solving a task bundle.
Although the cost of resolving each task may vary, 
we can refer to the average cost over a while,
because the tasks are universally sampled into task bundles.
For example, such cost can be dynamically measured from the reference BBHs' average cloud server cost of solving a task bundle.
Then, even if the likelihood of discovering a bug can be slim, the expected revenue of BBHs still covers the expected cost, as claiming only basic rewards can offset the computational cost.
Additionally, BBHs may obtain the premium bug reward, turning the expected return positive.

Moreover, this reward setup can deter BBHs who want to make profits by leeching only basic rewards, i.e., BBHs that run SMT solvers already used by reference BBHs make little contribution to the ecosystem.
The expected profit of a \BBH{} claiming only basic rewards is zero, because the expected cost of solving each task bundle is equal to the basic reward.
We formally prove no strategy can render the profit positive for such BBHs.

\begin{proof}
We assume that the cost of solving every task is the same, and the cost of solving the whole bundle is equal to the basic reward. 
Under these assumptions, we know that if BBHs solve the entire task bundle and claim only a basic reward, their expected profit is 0.
Then, because of task conversion, BBHs cannot determine the satisfiability without solving a task.
Therefore, the optimal strategy is to solve $t < n$ tasks and guess the rest.
We note that a BBH can only save cost by guessing a task is genuine (giving \unsat).

Let the basic reward be $\basicreward$, and the least cost of solving $t$ tasks is $\frac{t}{n}\basicreward$. 
If the guess is correct, then the profit is $\frac{n-t}{n}\basicreward$, otherwise, the profit is at most $-\frac{t}{n}\basicreward$.
Now we analyze the expected profit of this strategy. Let $P_{guess}$ represent the probability of guessing correctly. We require $\frac{n-t}{n}P_{guess}\basicreward - \frac{t}{n}(1-p_{guess})\basicreward \leq 0$, which implies $P_{guess} \leq \frac{t}{n}$.

As the probability of a task being genuine is $P$,  $P_{guess} = P^{n-t}$.
Note that $P$ is the parameter that can be set, so it is easy to find such a $P$ that satisfies $P^{n-t} \leq \frac{t}{n}$. One possible parameter is $\frac{1}{2}$.
In the scenario where $t=n-1$, if a BBH correctly solves $n-1$ tasks and gives \unsat to the last task, then $P = \frac{1}{2} < \frac{n-1}{n}$  holds for every $n > 2$.
For scenarios where $t$ ranges from $1$ to $n-2$, $\frac{1}{2}^{n-t} <\frac{t}{n}$ still holds for every $n > 2$. For every $n$, we can set $P$ to ensure that $P^{n-t} \leq \frac{t}{n}$ always holds true. 
Therefore, the BBH cannot make a positive profit by claiming only basic rewards.
\end{proof}

\subsubsection*{Quantitative Metric}
\name, as a verification service, should ensure the delivery of verification results based on quantitative metrics.
The binary's verification counts and upload time are recorded on the blockchain, serving as such metrics for users to quantitatively assess the trustworthiness of a target binary.
Besides, the \vc[s] and constraint files are also published, allowing users to reproduce the entire verification process.
Intuitively, higher verification counts and earlier upload time indicate better security.
We quantitatively analyze the threshold of the verification counts based on the following assumptions: the BV and BTM TEEs are not compromised, all BBHs are rational, all tasks are sampled into bundles with equal probability, the computational cost for solving every task is consistent, and a basic reward is equivalent to the cost of solving a task bundle.

Since different BBHs may adopt the same solving method, not all submissions enhance the verification results.  
A function $F$ is deemed \textit{effectively} verified $T$ times upon being verified by $T$ unique solving methods.
We define the verification count of a function $F$ to be $m$ when there are $m$ ``no-bug'' submissions of task bundles containing genuine tasks derived from $F$ (each $F$ is associated with several genuine tasks). 
Subsequently, we compute the probability that a specific function has been effectively verified for $T$ times after its verification count reaches $m$.

Assuming that $p$ of the ``no-bug'' submissions are from unique solving methods,
then the probability that an arbitrary function's ``no-bug'' submission is an effective verification becomes $p$.
We represent the expected effective verification times of a function as $T$.
Thus, employing the binomial distribution, we can deduce that the probability of a function being effectively verified at least $T$ times is $\prob{X \geq T}=1-\prob{X \leq T-1}=1-\sum_{k=0}^{T-1}{\binom{m}{k}}p^{k}(1-p)^{m-k}$ when its verification count is $m$.


Assuming a user believes that verifying a function at least 10 times ensures sufficient security, we consider an ideal scenario where 70\% of the ``no-bug'' submissions are from unique
solving methods. In this case, the probability that a given function is verified at least 10 times exceeds 99\% when its verification count is 21.

For a binary with $N$ functions, each with a verification count of $m$, the probability that all $N$ functions are verified at least $T$ times is at least $1 - N \left( \sum_{k=0}^{T-1} \binom{m}{k} p^{k} (1-p)^{m-k} \right)$, as deduced by applying the union bound. Consider a binary with $N = 200$ functions. If the verification count for each function reaches 28, the probability that every function of this binary is verified at least 10 times exceeds 99\%.

\subsection{Handling Existentially-Quantified Policies}
\label{sec:btm-sat}
The BTM's workflow for $\exists$-quantified policies is much simpler: 
\ignore{We briefly discuss their discrepancies here. 
For $\exists$-quantified policies, } it sends input constraint files directly to a task bundle without going through the Task Fabricator or the Task Converter. 
Once the constraints are solved with a \sat model and the model is validated by the \sat validator, the $\exists$-quantified policy is verified.

\ignore{
\hongbo{We plan to move Sybil attack to discussion?}
According to \autoref{fig:answer-tree}, with the BTM validating satisfiable models and being fully aware of fabricated tasks, the only potential attack lies in the basic reward case, processing \unsat results that are not directly validated.
While genuine tasks typically expect \unsat answers, undetected bugs could persist if both untrusted solvers and BBHs fail.
Moreover, a malicious BBH might launch a Sybil attack~\cite{douceur2002sybil} to proliferate multiple false identities, answering using reference solvers.
This would artificially inflate verification counts, misleading users into trusting the binary.
Nevertheless, as our reward structure effectively deters lazy BBHs, honest BBHs remain engaged.
Therefore, honest BBHs continue receiving task bundles, increasing the likelihood of detecting concealed bugs.
With the ability of honest BBHs to identify bugs and orthogonal countermeasures against Sybil attacks~\cite{maram2021candid}, any undetected bug becomes increasingly likely to be exposed over time.
Consequently, binary users can gauge the trustworthiness of verification outcomes based on verification counts and elapsed time since the binary's upload.
}

\ignore{
The bug bounty program relies on the BTM, so its availability may also harm verification soundness.
Binary users remain unaware of the bug's existence if the BTM is unable to handle submitted bugs.
To address the threats to BTM's availability (e.g., DoS attack), we institute a censorship-resistant submission channel here.
}

\ignore{
\subsection{Bug Bounty Protocol}
\label{sec:btm_workflow}

We present a simplified protocol followed by our bug bounty design here.
A detailed and formal description, along with relevant algorithms, can be found in~\autoref{appx:btm-full-workflow}.

\subsubsection*{Deployment of \name}
The BTM is hosted in TEE and managed by an untrusted administrator.
Upon initialization, the BTM produces a key pair and publicizes its public key.
Following this, the \textit{smart contract} (\smartcontract) is activated, integrated with the public key and the BTM's measurement, ensuring exclusive write privilege for the BTM to update \smartcontract's states.
The BTM and the BV establish a trusted channel after remote attestation.
This connection guarantees that the BTM receives the constraint files from a trustworthy source, and the BV can publish signed binaries on the blockchain.\hongbo{make the description consistent}

\subsubsection*{Task Generation and Publishment}
To begin with, a BBH requests a {\em task bundle} from the BTM (\ding{202}).
In response, the BTM creates the task bundle linked to specific constraint files.
This bundle encompasses genuine tasks (corresponding to functions to be verified and presumably to be \unsat) and fabricated tasks (modified from genuine ones via Task Fabricator but validated to be \sat).
Each task in the bundle is processed through the Task Converter, which converts original constraints to a series of equivalent ones~\cite{huth2004logic}, ensuring that
the original and converted constraints are both \unsat or agree on the same satisfiable models.
Since no polynomial-time algorithm can confirm the equivalence of two logical proposition sets (i.e., a series of constraints)~\cite{eiter2004simplifying, eiter2007semantical}, bypassing constraint solving is not feasible.
The finalized task bundle is then broadcasted on IPFS~\cite{ipfs} (\ding{203}) and the BBH can exclusively fetch it (\ding{204}).
Note that the BTM knows which tasks are fabricated and genuine and retains this mapping confidential.

\subsubsection*{Answer Submission}
Upon receiving the task bundle, the BBH can utilize any method to conduct constraint solving thanks to the open ecosystem.
The BBH submits answers for each task in the bundle (\ding{205})--either marking them as \unsat or furnishing a satisfiable model.
The BTM then validates the answers as described in \autoref{fig:answer-tree}.
For satisfiable models, the \sat Validator re-evaluates the constraints against this model; for \unsat answers, the BTM simply checks their alignment with the previously recorded genuine task.
The BTM then disperses the reward and updates verification results on the blockchain correspondingly (\ding{206}).

\subsubsection*{Verification Result Inquiry}
The binary users can query the \smartcontract about the verification results.
When users prepare to verify the binary, they iterate all functions in the binary and check the presence of bugs.
If not, they can further scrutinize verification counts and the upload date of the tasks, juxtaposing them against personal thresholds to evaluate the binary's trustworthiness.
Moreover, they can initiate remote attestation with the BTM to confirm an uncompromised BTM and consistent public key by comparing those recorded in the \smartcontract.
Any inconsistency is an alert of unsound verification.
Honest BBHs should also have scrutinized the values, indicating and announcing the untrusted BTM to the public if any.

\subsubsection*{Censorship-resistant Submission Channels}
Given the inherent availability concern of the BTM (e.g., system administrators are capable of suppressing bug submissions), an alternative censorship-resistant channel is essential.
In such scenarios, BBHs can instantiate a BTM locally within a TEE, exclusively activating the \sat Validator component.
Remote attestation subsequently vouches for the BTM's authenticity.
The \smartcontract{} cross-checks the attestation quote against the embedded BTM's measurement and verifies the signing certificate of the quote.
If these verifications succeed, the \smartcontract{} trusts this submission and marks the corresponding function as insecure.
}

\ignore{
\subsection{Security Analysis of the BTM}
\label{sec:btm-security}



\sen{only attack surface is case 3, the verification count is not trustworthy,  malicious attacks are to boost this step. Provide a time for user to refer, they will be more trustworthy than others. Lower bound is the distributed solver.}

The BTM needs to ensure soundness in the results and the sustainability of the bug bounty program.
We list the possible attack vectors from different entities that could undermine these two goals in~\autoref{tab:attack-vectors}.
The design of the censorship-resistant submission channel can effectively prevent DoS attacks resulting in soundness problems by the BTM administrator.
In this subsection, we will prove that the malicious BBHs cannot thwart our goals.

Note that a BBH can produce a correct answer by deciding the satisfiability of each task (e.g., using an SMT solver or guessing).
We define the expected profit of a BBH as the expected rewards minus the expected cost of calculating correct answers.
In the next lemma, we prove that the expected profit of only claiming basic rewards is 0. This suggests that lazy BBHs are unlikely to participate.
Given that BBHs are rational, they will not be lazy because no strategy can lead to a positive profit.



\begin{lemma}
For simplicity of computation, we assume the cost of solving every task is the same.
The expected profit of a \BBH{} claiming only basic rewards is 0 if the cost of solving each task bundle is equal to the basic reward. \hongbo{appendix}
\label{lemma:solving}
\end{lemma}

\begin{proof}
The full proof of~\autoref{lemma:solving} is given in~\autoref{sec:lemmaproof}. The main idea is that even when using an optimal strategy, BBHs cannot achieve a positive profit by solely claiming the basic rewards.
\end{proof}

\ignore{
\begin{lemma}
Assuming the probability of finding a bug is larger than 0 and the cost of solving the whole bundle is equal to the basic reward, the expected profit of honestly solving a task bundle is greater than 0.
\label{lemma:positive-profit}
\end{lemma}

\begin{proof}
As the honest BBHs will try to solve all tasks in the bundle, they can claim basic rewards, which are equal to their costs. Since the probability of finding a bug is greater than 0, the expected profit from claiming a bug bounty reward is also greater than 0.
\end{proof}

\hongbo{from Lemma 1, we should claim that the expected revenue for BBHs is positive, considering bug reward}
\sen{It may be another theorem or lemma, positive expected reward?}
}




Given that the basic reward offsets the cost of solving the entire task bundle, the expected profit for honest BBHs becomes positive when the probability of finding a bug exceeds 0, meaning that BBHs are capable of finding bugs.
Due to this positive expected profit, honest BBHs could be strongly motivated to participate~\cite{akgul2023bug}.

In the following~\autoref{thm:buggy-program-launderer}, we will show that the malicious BBHs cannot get a buggy binary to pass verification because of the participation of honest BBHs. Furthermore, the lower bound of the soundness guarantee aligns with the use of multiple trusted solvers.  As the Task Converter converts the constraint files, the malicious BBHs are unable to pinpoint the specific task to answer with \unsat. So their attack strategy would involve initiating a Sybil attack, generating multiple false identities, and each one answers using the untrusted SMT solver that \name employs\hongbo{maybe call it reference solver for short?}. 

\begin{theorem}
Assuming bugs exist, the probability of finding a bug is greater than 0, and honest BBHs participate in bug bounty programs.
If the cost of solving the entire bundle equals the basic reward, malicious BBHs cannot make a buggy binary pass verification.
\label{thm:buggy-program-launderer}
\end{theorem}

\begin{proof}
\autoref{lemma:solving} indicates that lazy BBHs will not participate in the bug bounty program.
Given that the probability of finding a bug is greater than 0 and honest BBHs will engage in the bug bounty program\quan{grammar error}.
As a result, malicious BBHs cannot prevent these honest BBHs from detecting bugs in buggy binaries.
\end{proof}
}

\section{Implementation}
\label{sec:implementation}

We implement the BV and part of the BTM, supporting the verification of SFI, IFC, and LVI side-channel mitigation policies.

\subsubsection*{The Binary Verifier}
The BV is developed in Rust with a few crates.
It employs the (dis)assembly framework, \texttt{iced-x86}, to disassemble binary code and model the semantics of instructions.
We use \texttt{Pest} as the \vc parser.
To reduce TCB, we independently implement CFG generation using \texttt{petgraph} and SSA transformation.

\ignore{
We also introduce optimizations to simplify the generated assertions and eliminate quantifiers in SFI enforcement and LVI mitigation policies.
For example, a predicated assertion $\tt FnPtr ~ rbx.1$ is replaced with an equality check on all possible values such that the predicate returns true, in the form of $\tt rbx.1 = FnPtr1000 \lor \tt rbx.1=FnPtr2000 \lor \cdots$.
Therefore, functions are no longer needed to model predicates.
}


\ignore{
The BV is compatible with additional analyzers and verifying some policies may require them.
For example, WASM binaries compiled by Lucet use the stack to transfer important abstract values (e.g., \texttt{HeapBase}), so we track the values stored on the stack.
This analyzer is integrated into the SSA transformation and assigns the values on the stack a version number (i.e., subscript).
In other words, we regard the stack addresses as similar to registers in SSA but indexed by the offsets.
Note this tracking is more accurate for compilers (e.g., Lucet) whose stack pointer adjustment amount is never determined dynamically
}



\ignore{
\vspace{3pt}\noindent\textbf{Constraint Generation}.
\hongbo{need help here}
To be compatible with various policies, we implement a policy \texttt{Matcher} trait that inserts assertions based on the assembly code, specifications, and hints.
We implement the policy \texttt{Matcher}s for LVI mitigation, which inserts \texttt{lfence} instructions after load instructions, and SFI for WebAssembly, which checks memory accesses are bounded and indirect calls are secure.
We leverage incremental solving in \texttt{smt2} file such that each assertion is checked in a dedicated stack and does not pollute the environment.
}

\ignore{
\vspace{5pt}\noindent {\it Type Checking.}
We also realize a type checker for the proof language, as the proof, semantics, and constraints share similar syntax but operate on different types of values.
The SMT solvers impose strict restrictions on the sort (i.e., type) of values that can occur in the constraints, so the BV must generate well-typed constraints.
This type checker also enables the BV to reject problematic proof at an early stage.
}

\subsubsection*{The Bounty Task Manager}
Realizing the full-fledged BTM is not technically novel but demands substantial engineering efforts (e.g., term rewriting system).
We only realize the unique components in \name, namely the smart contract and the \sat Validator, leaving the implementation of Task Fabricator and Task Converter as future work.
The smart contract is written in Solidity~\cite{solidity} and tested on Ethereum~\cite{wood2014ethereum}.
The BTM interfaces with Ethereum using \texttt{web3py}.
Binary users can query verification results from the smart contracts by leveraging Ethereum node services combined with the Web3 API.
We adopt IPFS~\cite{ipfs}, a decentralized storage solution, to house tasks, \vc[s], and certified binaries.
We leverage the open-source implementation by Automata~\cite{automata2025dcapattestation} to enable on-chain verification of attestations in our smart contract.

\subsubsection*{\vcgen[initcap]}
We prototype two \vcgen[plural] for SFI enforcement and LVI mitigation policies, plus the assertion generator for IFC enforcement is simply the unmodified \confllvm compiler.
The \vcgen for SFI is built on VeriWASM.
We tailored its source code, which conducts binary analysis to verify SFI based on abstract interpretation.
The \vc[s] are therefore generated using the abstract states.
Besides, we also implement a simple binary analyzer following the instrumentation strategy to mitigate LVI attack purposed by GNU assembler to generate \vc[s].
The analyzer produces \vc of load buffer alterations for all sensitive load and $\tt lfence$ instructions.
\begin{table}[!t]
  \begin{center}
  \caption{TCB size breakdown in LoC of \name, angr, VeriWASM, ConfVERIFY, Deflection, and Cedilla.}
  \label{tab:tcb}
    \begin{tblr}{
      hline{3-7,9-11} = {solid},
      hline{2,8,12} = {solid},
      vline{2,5,6} = {solid},
      columns = {c},
      cell{2}{1-4} = {r=4}{},
      cell{6}{1-4} = {r=2}{},
      cell{1}{5} = {c=2}{c},
      cell{2-Z}{5} = {}{l,teal!20!},
      cell{2-Z}{6} = {}{r,teal!20!},
      rows = {belowsep=0pt, abovesep=1pt, m},
      row{1} = {font=\bfseries},
      column{1} = {font=\bfseries},
    }
                          & {General\\Utilities} & {Binary\\Verifier} & {BTM/Solver/\\Checker}  & {Policy-specific\\Code}            &        \\
      \name               & 3.6K                 & 848                & 6.1K                      & SFI-VeriWASM                 & 625    \\
                          &                      &                    &                           & IFC-ConfVERIFY               & 335    \\
                          &                      &                    &                           & SFI-Deflection               & 500    \\
                          &                      &                    &                           & LVI                          & 77     \\
      angr                & 6.6K                 & 67K                & 532K                      & HeapHopper~\cite{heaphopper} & 2.1K   \\
                          &                      &                    &                           & MemSight~\cite{memsight}     & 3.8K   \\
      VeriWASM            & 2.5K                 & 984                & -                         & SFI                          & 1.9K   \\
      ConfVERIFY          & \mytilde 40K         & -                  & -                         & IFC                          & 1.5K   \\
      Deflection          & \mytilde 9.6K        & -                  & -                         & SFI                          & \mytilde700   \\
      Cedilla             & \mytilde 16.5K       & 4.0K               & 1.4K                      & Type-Safety                  & 3.2K   \\
    \end{tblr}
  \end{center}
\end{table}

\section{Evaluation}
\label{sec:evaluation}

In this section, we first compare the TCB of \name with other binary verification tools.
Then, we evaluate the performance of the BV and constraint solving.
Finally, we present the estimated cost of running the bug bounty protocol on blockchain.


\subsection{TCB Size Analysis}
\label{sec:tcb-analysis}



\subsubsection*{Breakdown}
\autoref{tab:tcb} details the TCB size of the components in \name and 4 other representative tools or frameworks.
\spchanges{\emph{General Utilities} is the backbone of a verifier, including intermediate representation, encoding of semantics, and binary utilities.}
\emph{Binary verifier} is the tiny core of verification, thanks to the innovative design of \name that offloads the most sophisticated analyses.
Counting utilities, the BV (excluding policy-related code) is roughly 4.5 KLOC plus dependencies.
The BTM removes trust in SMT solvers.
Since a full-fledged BTM is not realized (see \autoref{fig:btm-workflow}), we estimate the TCB of all components, including the unimplemented Task Converter based on term rewriting~\cite{trs-rs} and Task Fabricator, totaling to \mytilde 6 KLOC.
Note that the \vc parser and disassembler are excluded from the TCB, because incorrect \vc is rejected in \vc validation, and \name certifies the re-assembled binary, rather than the original binary being disassembled. 
Taking advantage of formally verified ELF builders~\cite{compcertELF}, the assembler can be excluded from the TCB.


\subsubsection*{Comparison with Other Systems}
We compare the TCB size with VeriWASM, \confllvm, angr, and Cedilla, as broken down in~\autoref{tab:tcb} with policy-specific code size shaded.
The component breakdowns are roughly mapped to each other with our best effort, given that each framework uses different techniques and code modularization is often imperfect.

Compared with angr~\cite{angr}, a binary analysis framework based on symbolic execution, the TCB of \name is significantly smaller in all aspects.
Although TCB under general utilities is not very meaningful since angr models more instructions than \name does, TCB in other aspects is still important.
We cannot find direct SFI, IFC, and LVI mitigation via angr. However, implementing additional policies on angr likely will require more code than the counterpart on \name, according to the examples\cite{heaphopper, memsight}.

In VeriWASM, its policy-specific code is roughly three times the counterpart in \name, suggesting the extensibility of \name. 
Overall, the TCB of VeriWASM and \name under general utilities and binary analysis are comparable, even though VeriWASM supports only WebAssembly SFI and it would be strenuous to extend to new policies.
The strength of VeriWASM is that it requires no constraint solving.
However, VeriWASM does not achieve an open and auditable verification service, and we note that this is in general a trade-off between generality and complexity.

\confllvm is accompanied by a verifier called \confverify~\cite{confllvm}, which verifies IFC in binaries compiled by \confllvm. The TCB of \confverify includes itself (1.5 KLOC) and a disassembler (\mytilde 40 KLOC) used for CFG construction and semantics lifting. \name-based verifier only requires 335 lines of additional policy-specific code, much smaller than \confllvm's code excluding disassembler. Moreover, since IFC requires a \sat model, Task Converter and Task Fabricator are unneeded (see \autoref{sec:btm-sat}), further reducing its overall TCB size to be close to \confverify's.


\spchanges{
Deflection~\cite{sgx:deflection} verifies enclave binaries emitted from a modified compiler with SFI enforcing instrumentation, thus permitting untrusted binaries to be executed via a customized loader after rewriting.
We extract its verified policy in \name with 500 LoC.
The instrumentation strategy allows Deflection to verify SFI via pattern matching, thus resulting in a small TCB similar to \name's counterpart.
}

The PCC-based Java certifying compiler, Cedilla~\cite{colby2000certifying}, also features a relatively small TCB size, but without support for versatile policies.
It still relies on a complex verification condition generator, consisting of a program analyzer (symbolic executor), a policy plug-in for Java constructs, and utilities for a built-in type system.



\subsubsection*{TCB Size Reduction Revisit}
The BV and the BTM successfully offload binary analysis and constraint solving, respectively. 
The BV allows binary analysis to be highly complex but excluded from the TCB.
Moreover, \name's support for multiple \vcgen[plural] reduces the TCB size substantially in verifying different policies.
On the other hand, the BTM gains more than one order of magnitude of TCB size reduction (typical SMT solvers contain 200K to 500K LOC~\cite{de2008z3, dutertre2014yices, barbosa2022cvc5}).
Besides, \name not only creates an open verification ecosystem, but also provides a verification service with auditability and reproducibility, which are not simultaneously achieved by any other tools.

\subsection{Evaluation of SFI}
\label{sec:sfi-eval}




We conduct the evaluation on a 32-core@2GHz server for \name's version of VeriWASM.
We employ the same compilation toolchain and SPEC 2006 benchmark as VeriWASM used, compiling the WASM modules into x86 ELFs.
We compile ten programs at \texttt{-O3} and \texttt{-O0} optimization levels and test \name on 20 binaries.
To ensure that \name generates verifiable constraints and evaluate the performance of \name, we use CVC4 1.8, CVC5 1.0.5, and bitwuzla 1.0 to solve constraints and measure the time elapsed.

\subsubsection*{Verification Result}
\name's results align with VeriWASM, as all functions verified by VeriWASM also pass \name's.
However, some functions cannot be verified by either system, even though \cite{veriwasm} states that all binaries in SPEC 2006 can be verified. 
We speculate that the inconsistencies are due to slight variations in the compilation toolchains.



\subsubsection*{Execution Time}
The execution time of \vc generation and the BV workflow is roughly proportional to the size of the verified function.
The BV runs 9.4 times faster than \vc generation, which mainly performs offloaded binary analysis.
The BV processes 31K instructions per second.
We believe this performance is sufficient to handle large software.
However, constraint solving is significantly slower, especially for certain functions.
In our benchmarks, 99\% of functions contain less than 8171 instructions, and constraint solving using CVC4 is completed in 991 seconds\footnote{CVC4 timeouts on the largest function \texttt{encode\_one\_macroblock} of 63049 instructions in \texttt{h264ref}. We exclude them from the evaluation. }.



\begin{figure}[!t]
	\centering
	\begin{subfigure}[t]{0.49\textwidth}
		\centering
		\includegraphics[width=\textwidth]{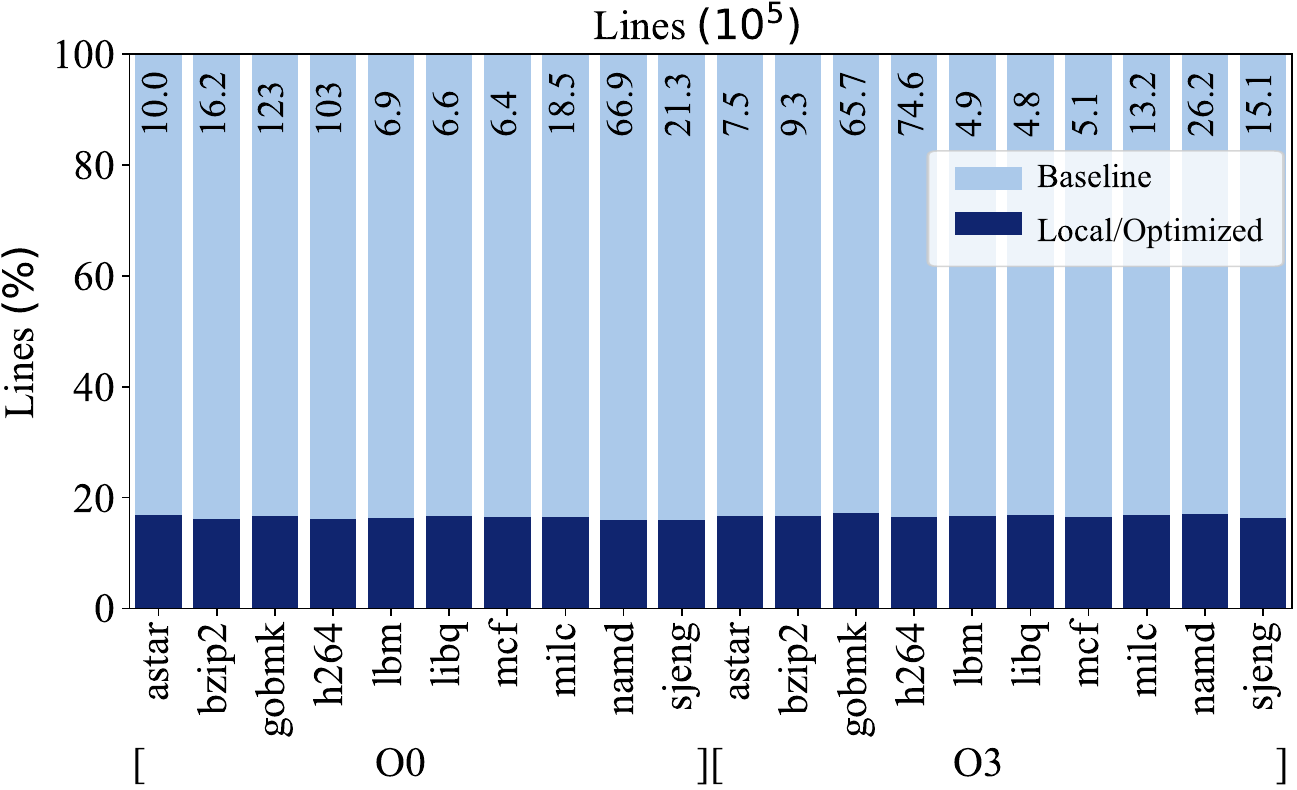}
            \Description{Description placeholder}
		\caption{Constraint file (task) size.}
		\label{fig:eval-lines}
	\end{subfigure}
	\hfill
	\begin{subfigure}[t]{0.49\textwidth}
		\centering
		\includegraphics[width=\textwidth]{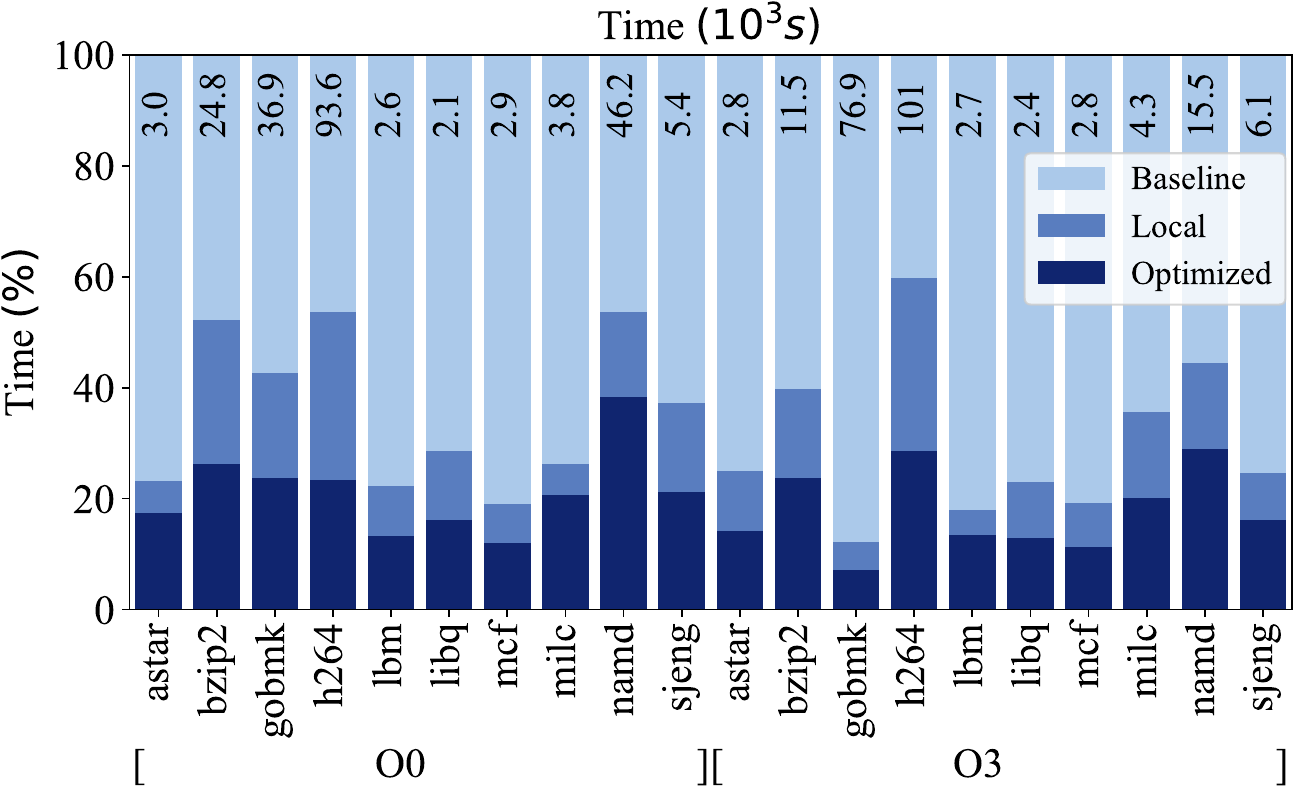}
            \Description{Description placeholder}
		\caption{Average constraint solving time.}
		\label{fig:eval-time}
	\end{subfigure}
        \Description{Description placeholder}
	\caption{Performance evaluation of the Binary Verifier on SPEC 2006.}
	\label{fig:eval}
\end{figure}

\subsubsection*{Constraint Solving}
Now, we evaluate the performance of our design (\textsc{Local}), comparing it with the \textsc{Baseline} setting without instruction-level \vc validation and optimizations.
From \autoref{fig:eval-lines}, the constraint file sizes of \textsc{Baseline} and \textsc{Local} show that local validation effectively removes semantics irrelevant to proving policy compliance.
This achieves an average of 84\% and 83\% task size reduction on \texttt{-O0} and \texttt{-O3} optimizations flags, respectively.
\textsc{Local} and \textsc{Optimized} have identical sizes due to only eliminating the disjunctions when checking memory access. 
We also assess the performance of constraint solving on different solvers.
\autoref{fig:eval-time} demonstrates the geometric mean solving time of 3 solvers for each binary.
We also include \textsc{Optimized} setting where the memory access category hint is enabled.
On average, local validation reduces 
solving time by 70\%, and hint optimization further reduces the time by 40\% (82\% overall reduction).
These reductions in task size and solving time indicate that our design can effectively prune policy-agnostic semantics and enhance performance.





\ignore{
Compared to the baseline version with no localized proof validation, the optimization reduces the constraint file size and solving time by roughly 28.02\% and 32.04\%, respectively on average.
The unoptimized BV generates 32,255,386 and 16,068,252 lines of constraints for all tests in SPEC 2006, respectively at compiler optimization level \texttt{-O0} and \texttt{-O3}.
Whereas the optimized BV generates 23,333,226 and 11,449,893 lines.
The proof file size reduces 27.66\% and 28.74\% for \texttt{-O0} and \texttt{-O3}, respectively.
\autoref{fig:optimization} demonstrates the constraint solving times for different SMT solvers, categorized by the compiler optimization levels.
We observe an average solving time reduction at 33.05\% and 30.78\% for \texttt{-O0} and \texttt{-O3}, respectively across all solvers, in which XXX has the best performance and XXX receives the greatest performance improvement.
}

\if lvi
\subsection{LVI Mitigation}

We also evaluate \name on Lfence policy, which verifies LVI mitigation.
The experiment setting is identical to the previous one, with the exception of switching the compilation toolchain to GNU \texttt{gcc/g++} 9.4 and \texttt{as} 2.36.1.
This policy is much simpler than LucetSFI: we implement the policy in 77 LOC.
Extending the semantics to incorporate the behavior of load buffers adds less than 100 LOC to the PC.
This minor modification effort exemplifies the extensibility and flexibility of \name.

We evaluate the BV and conduct constraint solving on the binaries from SPEC 2006.
Since the proof only contains relations of the load buffer, both proof validation and constraint generation are much simpler and faster than for LucetSFI.
All binaries are successfully verified in less than 870 seconds.
\fi

\subsection{Evaluation of IFC and LVI Mitigation}
\label{sec:lvi-evaluation}

For IFC and LVI Mitigation, we validate the correctness of \name-based verifier. The experiment setting is identical to the previous one, except that the binary being verified is generated by \confllvm and GNU \texttt{gcc/g++} 9.4 and \texttt{as} 2.36.1 (with option \texttt{-mlfence-after-load=yes}) for IFC and LVI Mitigation, respectively.
We run the BV and conduct constraint solving on the binaries from NGINX and OpenLDAP for IFC, and SPEC 2006 for LVI. 

All of the binaries are successfully verified by \name-based verifiers. Besides, we intentionally tampered with binaries for IFC check to maliciously leak secret information via modifying random magic sequences as well as segment registers. \name-based IFC verifier successfully identifies the violations and rejects those tampered binaries.

\ignore{

}
\subsection{Bug Bounty Protocol Cost Estimation}

We first evaluate the gas cost of deploying and interacting with the smart contracts.
We use Solidity v0.8.28 to compile our smart contracts and deploy them on the Sepolia Testnet~\cite{ethereumnetworks}.
The gas cost for deploying the BTM contract is 2,219,771, while it stands at 214,758 for publishing a task bundle, and 33,862 for recording the verification result. Regarding the censorship-resistant bug submission, the gas cost of verifying an ECDSA attestation quote is 5,546,083. The gas cost of queries is 0 for binary users.

Then we evaluate the cost of using the bug bounty protocol for developers who want all functions to be verified with high confidence (e.g., the probability that each function is verified at least 10 times by honest BBH is 99\%).
We create a task bundle including 200 tasks.
According to experiments, it takes 20 minutes on average to complete this bundle using a single CPU (c.f. \textsc{Optimized} in \autoref{sec:sfi-eval}). Considering the AWS instance with 4 CPUs and 8GB memory, which costs \$0.254 hourly, we set the basic reward at \$0.021.

Based on the quantitative analysis in~\autoref{sec:trustworthy-verification},
if $p$ of the ``no-bug'' submissions are from unique solving methods, $q$ of the tasks in a task bundle are genuine, then the cost of basic rewards
is depicted in~\autoref{fig:btm-cost}.
In the ideal scenario discussed in~\autoref{sec:trustworthy-verification}, the protocol costs \$1.18 to verify a binary with 200 functions.
Even in an extreme scenario where only 10\% of ``no-bug'' submissions are associated with unique solving methods, the protocol costs \$11.05 to achieve the same outcome.
The cost to submit a bug using the censorship-resistant channel is \$29.4 as of February 2025, in Ethereum.
The total transaction costs can be reduced to less than \$0.01 by considering an alternative to Ethereum, such as Arbitrum~\cite{arbitrum}.

\begin{figure}[!tbp]
	\centering
	\includegraphics[width=0.6\textwidth]{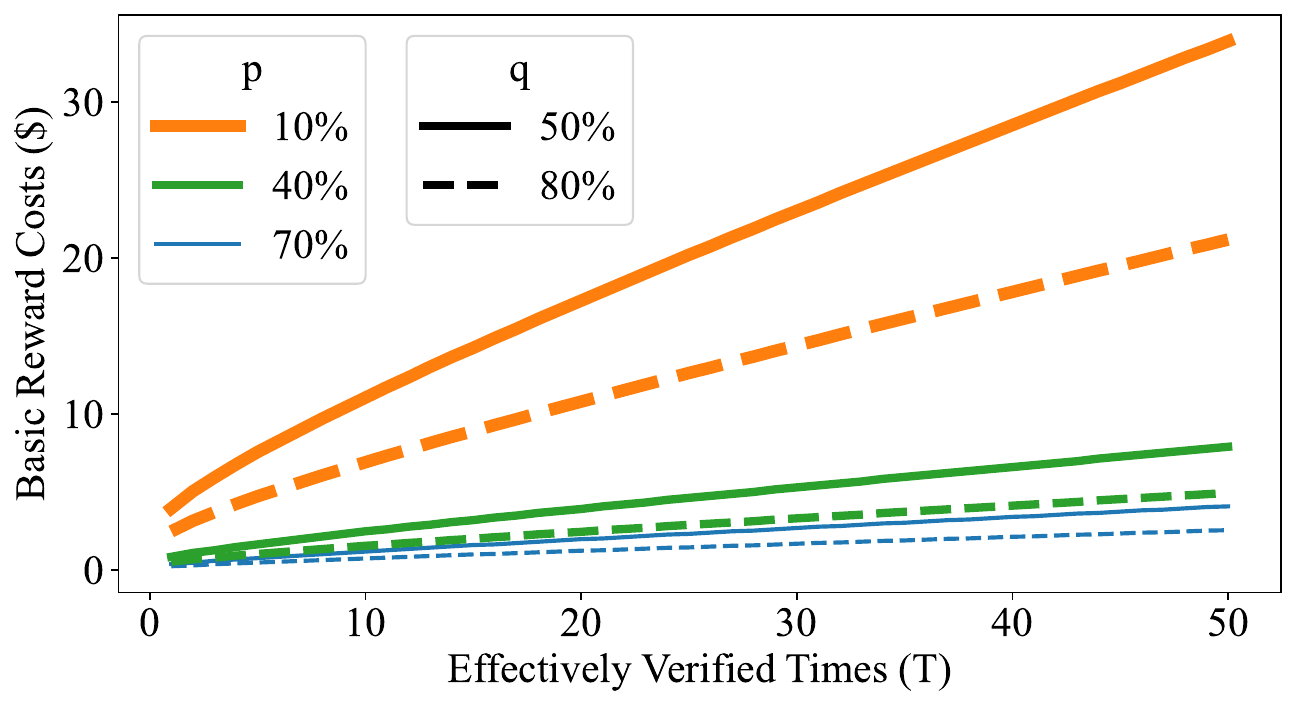}
        \Description{Description placeholder}
	\caption{Cost of basic rewards to ensure that 200 functions are effectively verified for T times over 99\% likelihood.}
\label{fig:btm-cost}
\end{figure}

\section{Limitation and Discussion}

\subsubsection*{Comprehensive Implementation}
As a prototype, \name is not yet integrated with a formally verified assembler (e.g., CompCertELF~\cite{compcertELF}) to generate verified binaries.
Furthermore, the BTM also lacks a complete implementation of the Task Fabricator and Task Converter (see \autoref{fig:btm-workflow}) due to the significant engineering effort required.
We posit that designs from similar systems, such as term rewriting systems~\cite{trs}, could be adapted for this purpose.

\subsubsection*{Architectural Extensions}
\name's frontend uses iced-x86's IR as an interface to binaries, and the rest of Agora is architecture-independent. To extend \name's frontend to other architectures such as ARM and RISC-V, which are prominent platforms for emerging TEEs~\cite{arm:cca, riscv:penglai},
we can either lift their assembly code to our IR, or employ existing multi-arch disassembler toolchains (e.g., Capstone~\cite{capstone_web} and yaxpeax~\cite{yaxpeax}).
Because they are based on RISC ISA design, we believe that supporting ARM and RISC-V is less demanding than supporting x86.
Since the disassembler is untrusted in Agora, an architectural extension would not substantially affect the overall TCB of Agora. 

\subsubsection*{Reliability of \unsat Result} With \sat Validator in BTM, solvers return \sat to signify security compliance of may policies can be formally validated. However, as mentioned in \autoref{sec:policy}, for must policies, \name cannot formally validate \unsat results from solvers due to the challenge of validating \unsat results. While the SMT community has made significant efforts to utilize \unsat cores as verifiable proofs, a unified proof format with broad support across different theories remains lacking~\cite{besson2011flexible, hoenicke2022simple, andreotti2023carcara, czerner2024resolution}. Instead, \name crowdsources the solving tasks to BBHs, which provide a statistical security guarantee (\autoref{sec:trustworthy-verification}). That is, the confidence in \unsat results—and thus in the software's policy compliance—grows as more independent BBHs corroborate the finding. To address the concern that an insufficient number of BBHs can undermine the reliability of these results, \name's operator can deploy ``reference BBHs'' that run a diverse set of well-established SMT solvers for a minimum level of statistical security guarantee.
Our future work includes designing a dedicated \unsat verifier with a small TCB size or integrating an emerging one into \name.



\ignore{
\name currently has several limitations that we plan to address in the future.
Firstly, our implementation of the BV is not yet integrated with a formally verified assembler (e.g., CompCertELF~\cite{compcertELF}).
Secondly, the current modeling of the semantics is coarse-grained since the BV approximates some semantics using the information from \texttt{iced-x86}.
While this does not affect soundness, it may reject \vc[s] containing lines for instructions without detailed semantics modeling.
Lastly, \name currently only supports x86 binaries, but our design could be extended to other architectures.
}

\subsubsection*{Scalable Deployment}

Several opportunities exist to improve scalability and reduce costs for large-scale verification.
First, since verification tasks are typically independent, performance can be scaled by running multiple BV and BTM instances in parallel.
The open design of \name also encourages community contributions in the form of novel constraint-solving algorithms and efficient program analyses, which can further facilitate the verification of large systems.
Second, to manage transaction costs induced by the blockchain, which scale linearly with the number of verified binaries, smart contracts can be deployed on Ethereum Virtual Machine-compatible blockchains with lower fees, such as Arbitrum~\cite{arbitrum}. Such a migration would require no modifications to the smart contracts while benefiting from significantly reduced transaction costs.

\section{Conclusion}
\label{sec:conclusion}
In this work, we introduce \name, a verification service requiring minimal trust from users while open to versatile policies, untrusted contributions, and public audits.
Our design effectively offloads the heaviest components of program verification, namely the static analyzers and SMT solvers, to the untrusted domain.
Additionally, we implement a prototype of the BV and conduct evaluations on its verifiability and performance when checking SFI, IFC, and LVI mitigation policies.
We show \name's suitability for diverse verification tasks of confidential computing on the cloud.

\section*{Data-Availability Statement}
As explained in \autoref{sec:implementation}, we implemented the majority of \name, including the Binary Verifier as well as the smart contract and bug validator of the Bounty Task Manager.
Our artifact can be found on GitHub: \url{https://github.com/ya0guang/agora} or Zenodo~\cite{aedoi}.
\name requires hardware TEEs and participants to achieve its claimed security guarantees.
However, evaluating our artifact does not necessarily depend on them.
We plan to submit \name for Artifact Evaluation and make it available to the community.

\ignore{
\section{Research Ethics}
\label{sec:ethics}

The authors of \name carefully reviewed related documents and examined our work. 
In its design and implementation, \name is a generic verification framework aimed at proving the safety properties of programs. As a verification tool, we believe that there is no ethical concern that lies in \name.

However, one of the verification framework's job is to uncover new bugs in programs. During our case study, we identified a few unexpected events when verifying one of the security policies (see discussion in \autoref{sec:sfi-eval}). While these problems do not incur any security concerns on real-world programs (thus, no CVEs should be raised), we reported our findings to the authors of the corresponding work. 

\section{Open Science}
\label{sec:openscience}

To support the open science policy, the authors are willing to open source the implementation of \name once the paper is accepted. Any dataset and evaluation results of corresponding case studies (excluding those under licensing restrictions) will also be shared.
}

\begin{acks}
We would like to express our gratitude to the anonymous reviewers for their insightful feedback and suggestions. 
We also appreciate the insights from Prof. Ning Luo.
This research was sponsored by National Science Foundation grants CNS-2207231, CNS-2401496, CNS-2401182, CNS-2207214 and a gift from Intel. 
Any opinions, findings, and conclusions or recommendations expressed in this material are those of the author(s) and do not necessarily reflect the views of the National Science Foundation and Intel Corporation.
\end{acks}

\appendix

\newcommand{\onrecv}{\textcolor{mygreen}{\textbf{On receive}}}
\newcommand{\oninit}{\textcolor{mygreen}{\textbf{On initialization}}}
\newcommand{\func}[1]{\textcolor{idealfun}{\text{#1}}}
\newcommand{\msg}[1]{\textcolor{myviolet}{\text{``#1''}}}
\newcommand{\cmt}[1]{\textcolor{mygrey}{\text{// #1}}}

\label{sec:appendix}

\section{The Workflow of BTM}
\label{appx:btm-full-workflow}

In this section, we elaborate on the detailed workflow of the BTM. We formally outline the \smartcontract's functionality in~\autoref{fig:smart-contract} and the BTM TEE's workflow in ~\autoref{fig:btm-protocol}.

\begin{figure}[!t]
\protocol{Smart Contract}{
\onrecv{} ($\msg{Create}, \pkbtm, \btmhash{}$): \\
\quad Store $\pkbtm$ and $\btmhash{}$ \\
\quad $\functions \gets \{\}$, $\taskbundles \gets \{\}$ \\
\onrecv{} ($\msg{PublishTaskBundle}, \text{msg}, \text{sig})$: \\
\quad \pcif{$\verifySignature(\text{msg}, \text{sig},\pkbtm{})$} \\
\quad\quad $\bundlehash{}, B, \addrbbh, deadline \gets \text{msg}$ \\
\quad\quad $\taskbundles[\bundlehash{}] \gets (B, \addrbbh{}, deadline, \false, \bot, \bot)$ \\
\onrecv{} ($\msg{ClaimReward}, \text{msg}, \text{sig}, \text{quote})$: \\
\quad\pcif{$\text{sig} \neq \bot$}  \\
\quad\quad $verified \gets \verifySignature(\text{msg}, \text{sig},\pkbtm{})$ \\
\quad \pcelse \cmt{censorship-resistant submission channel} \\
\quad\quad $verified \gets \verifyAttestation(\text{msg}, \text{quote}, \btmhash{})$   \\
\quad $\bundlehash{}, \addrbbhbtm{}, reward, \verifiedtasks, \buggytasks \gets \text{msg}$ \\
\quad $\_, \addrbbhsc{}, deadline, \_,\_,\_ \gets \taskbundles[\bundlehash{}]$ \\
\quad \pcif{$verified \land (\blocktime{} > deadline \lor 
  \addrbbhsc{} = \addrbbhbtm{})$}  \\
\quad\quad pay $reward$ to $\addrbbhbtm{}$ \\
\quad\quad $\taskbundles[\bundlehash{}] \gets (\_,\_,\_, \true,\verifiedtasks, \buggytasks)$ \\
\onrecv{} ($\msg{UpdateFunctions}, \text{msg}, \text{sig}$): \\
\quad \pcif{$\verifySignature(\text{msg}, \text{sig},\pkbtm{})$} \\
\quad\quad $\updatedfunctions \gets \text{msg}$ \\
\quad\quad traverse $\updatedfunctions$ to update $\functions$ \\
\onrecv{} ($\msg{QueryFunction}, \functionhash$): \\
\quad\quad \pcreturn $\functions[\functionhash{}]$ \\
\onrecv{} ($\msg{QueryTaskBundle}, \bundlehash{}$): \\
\quad\quad \pcreturn $\taskbundles[\bundlehash{}]$
\quad
}
\Description{Description placeholder}
\caption{The pseudocode of the \smartcontract{}. $\verifySignature{}$ use the embedded $\pkbtm{}$ to verify signatures and $\verifyAttestation{}$ use the embedded $\btmhash{}$ to verify TEE quotes from remote attestation. $\blocktime$ is a blockchain built-in variable representing the current time.}
\label{fig:smart-contract}
\end{figure}

\begin{figure}[!htbp]
\centering
\protocol{BTM TEE}{
\oninit{}:\\
\quad $\skbtm \sample \bin^{256}$ \cmt{signing key of the BTM or $\bot$} \\
\quad derive $\pkbtm$ from $\skbtm$ \\
\quad initialize $M$ \cmt{task to function mapping}\\
\onrecv{} ($\msg{ValidateSubmittedAnswer}$, \\
\qquad $(a_1, \dots, a_n), \bundlehash{}, \addrbbh{}$) \cmt{input from BBH} \\
\quad fetch task bundle $B=(T_1, \dots, T_n)$ from SC using $\bundlehash{}$ \\
\quad $(g_1, \cdots, g_n) \gets \determine(B)$ \\
\quad $reward \gets 0, \verifiedtasks \gets [], \buggytasks \gets []$ \\
\quad \pcfor{$i \gets 1$ \textbf{to} $n$} \\
\qquad \pcif{$g_i = \false$} \\
\quad\qquad \pcif{$a_i = \unsat \lor \neg\validate(a_i, T_i)$} \\
\qquad\qquad $reward = 0$; break \\
\quad\qquad \pcelse ~ $\verifiedtasks.push(T_i)$ \\
\quad \pcif{$reward \neq 0$} \\
\qquad \pcfor{$i \gets 1$ \textbf{to} $n$} \\
\quad\qquad \pcif{$g_i = \true \land a_i \neq \unsat$} \\
\qquad\qquad  \pcif{$\validate(a_i, T_i)$} \\
\quad\qquad\qquad $\buggytasks.push(T_i)$ \\
\quad\qquad\qquad $reward = reward + \bugreward$  \\
\qquad\qquad \pcelse ~ $reward = 0$; break \\
\quad \pcif{$reward \neq 0$} \\
\qquad $\text{msg} \gets (\bundlehash{}, \addrbbh{}, reward, \verifiedtasks, \buggytasks)$ \\
\qquad \pcif{$\skbtm \neq \bot$} \\
\quad\qquad use $\skbtm{}$ to sign msg and get sig \\
\qquad \pcelse ~ $\text{sig} \gets \bot$ \\
\quad conduct attestation and get quote \\
\quad send (msg, sig, quote) to \BBH{} \\
\onrecv{} ($\msg{UpdateVerificationResult}$, $(H_{{\text{B}_{0}}}, \dots, H_{{\text{B}_{m}}})$): \\
\quad $\updatedfunctions \gets \{\}$ \cmt{verified task bundle hashes}\\
\quad \pcfor{$i \gets 1$ \textbf{to} $m$} \\
\qquad fetch $\verifiedtasks_i$, $\buggytasks_i$ from SC using $H_{{\text{B}_{i}}}$ \\
\qquad \pcfor{$T$ \textbf{in} $\verifiedtasks_i$} \\
\quad\qquad $F \gets M[T]$; update $\updatedfunctions$ with $F$ \\
\qquad \pcfor{$T$ \textbf{in} $\buggytasks_i$} \\
\quad\qquad $F \gets M[T]$; update $\updatedfunctions$ with $F$ \\
\quad \pcif{$\skbtm \neq \bot$} \\
\qquad $\text{msg} \gets \updatedfunctions$ \\
\qquad use $\skbtm{}$ to sign msg and get sig \\
\quad \quad send ($\msg{UpdateFunctions}$, msg, sig) to SC
}
\Description{Description placeholder}
\caption{The pseudocode of \BTM{} TEE}
\label{fig:btm-protocol}
\end{figure}

\subsection{System Setup}
\label{appx:systemsetup}
The system is deployed and managed by an untrusted administrator, and other entities are illustrated in~\autoref{fig:btm-workflow}.
The BTM deployed on TEE consists of three components: Task Fabricator, Task Converter, and \sat Validator.
The smart contract (SC) deployed on the blockchain records verification results and facilitates reward payments.

After attestation, the BTM and the BV establish a secure channel, ensuring that the BTM only takes input from the trusted BV.
Then, the BTM generates a key pair ($\skbtm{},\pkbtm{}$), and the SC records its $\pkbtm{}$.
This allows the SC to verify that the BTM sent a message.
The SC is programmed so that only the BTM can create and update bug bounty tasks.
Besides, the SC also records the measurement of the BTM ($\btmhash{}$) to verify attestations (c.f.~\autoref{sec:censorship-resistant-bug-submission}).

\subsection{Generating and Publishing Tasks}

To begin with, BBHs request {\em task bundles} from the BTM (\ding{172}).
Once the BTM receives a request, it generates a task bundle $B=(T_1, \dots, T_n)$ where $T_i$ is a fabricated or genuine task with probability $Q$.
To generate a genuine task, the BTM randomly picks a set of constraints that correspond to a function to be verified.
To generate a fabricated task, the Task Fabricator removes certain constraints from the original constraint files to deliberately introduce bugs.

Both the genuine and fabricated tasks undergo conversion by the Task Converter.
Based on term rewriting system~\cite{trs}, the Task Converter converts a set of constraints to a series of equivalent constraints~\cite{huth2004logic} (for a constraint set $P$, $\bigwedge P \leftrightarrow \bigwedge\mathtt{mutate}(P)$), which \textit{always} have the same satisfiability).
Note that the constraints can be regarded as logical propositions, and there is no polynomial time algorithm to determine if two sets of logical propositions are equivalent~\cite{eiter2004simplifying, eiter2007semantical}.
Hence, solving past tasks does not aid BBHs in solving new tasks, and BBHs cannot determine if a task is genuine (i.e., $\unsat$) without solving it.

For any given function $F$, a task bundle includes at most one genuine task derived from $F$.
The BTM maintains a mapping $M$ from genuine tasks to the functions from which they are derived.
This mapping is concealed from the BBH as $M$ may help \BBH{}s to give \unsat answers without conducting constraint solving.

Each task bundle $B$ has a corresponding Boolean array $G_B=(g_1, \dots, g_n)$ stored in the BTM, recording if $T_i$ in $B$ is genuine.
The function $\determine(B)$ returns the stored corresponding array $G_B$.

The BTM publishes the task bundle $B$ with its hash $\bundlehash{}$ on the blockchain and sets a deadline before which only the requester can submit the solutions for $B$ (\ding{173}). Any unresolved task bundle becomes available to all BBHs after the time threshold, preventing malicious \BBH{} from hoarding task bundles but not solving them. The task bundle is recorded in the SC after verifying the BTM's signature (\msg{PublishTaskBundle} in \autoref{fig:smart-contract}).

The \BBH{} can now fetch $B$ using $\bundlehash{}$ from the SC (\ding{174}) and work on it to produce and submit an {\em answer} to the SC to claim rewards, as we discuss next.

\subsection{Answer Submission}

For a given task bundle $B=(T_1,\dots,T_n)$, a BBH needs to produce an {\em answer} in the form of $A=(a_1, \dots, a_n)$ where $a_i$ is either $\unsat$ result or a satisfiable model to $T_i$.
The BBH submits the answer $A$, the task bundle hash $\bundlehash{}$, and her address $\addrbbh{}$ (to receive reward payment) to the BTM (\ding{175}) for answer validation.

As described in~\autoref{fig:btm-protocol} (\msg{ValidateSubmittedAnswer}), once the BTM receives the answer validation request, it first fetches task bundle $B$ from the SC using $\bundlehash{}$. The function $\validate(a_i, T_i)$ returns true if and only if $a_i$ satisfies $T_i$. 
An answer is {\em correct} if all fabricated tasks are solved correctly. 
BBHs who submit a correct answer will receive a basic reward of $\basicreward{}$. Additionally, if the answer includes a satisfying model for a genuine task, it will be rewarded with a bug reward of $\bugreward{}$. Finally, the BTM returns the signed validation result to the \BBH{} (\ding{176}).

The \BBH{} sends the validation result to the SC. Once the signature is verified, the SC pays the BBH rewards and labels the verified and buggy tasks in the task bundle (see~\autoref{fig:smart-contract}, \msg{ClaimReward}). 

As described in~\autoref{fig:btm-protocol} (\msg{UpdateVerificationResult}), the BTM updates the verification results on the blockchain after answer validation. The SC handles the update request after verifying the BTM's signature (described in~\autoref{fig:smart-contract}, \msg{UpdateFunctions}). For each function $F$ to be verified, the SC stores the counts of submissions $F$ that have been verified and any confirmed bugs in $F$.

\subsection{Querying the Result}
Binary users query the smart contract for verification results and compare the obtained verification counts (\ding{177}) with their threshold counts to determine binary's trustworthiness.


\subsection{Censorship-resistant Submission Channels}
\label{sec:censorship-resistant-bug-submission}

We provide a censorship-resistant submission channel as an alternative when the official BTM TEE is unavailable.
In this case, the \BBH{}s run the BTM in a local TEE.
The locally deployed BTM can attest its identity to the \smartcontract{} via remote attestation rather than using a signature.
The only difference in \autoref{fig:btm-protocol} (\msg{ValidateSubmittedAnswer}) is that $\skbtm$ is set to $\bot$, resulting in an empty signature.
When the SC receives the verification result, it validates the result by comparing the attestation quote with the previously embedded measurement ($\btmhash$) and verifying the certificates signed for the quote, confirming the BTM TEE used by the BBH is authentic.
This guarantees that bugs can be reported even if the BTM is unavailable or deliberately refusing service.

\ignore{
\autoref{fig:eval-lines} showcases the constraint file sizes of binaries from SPEC 2006 under different settings, whereas 
\autoref{fig:eval-time} exhibits the average constraint solving time.}

\bibliographystyle{ACM-Reference-Format}
\bibliography{ref}

\end{document}